\documentclass[10pt, conference, letterpaper]{IEEEtran}
\ifCLASSINFOpdf
\usepackage[pdftex]{graphicx}
\else
\usepackage[dvips]{graphicx}
\fi

\usepackage[pass,]{geometry}

\usepackage{amsmath,amsthm,amssymb,amsfonts}
\usepackage{epsfig}
\usepackage[tight]{subfigure}
\usepackage{graphicx}
\usepackage{cite}
\usepackage{times}
\usepackage{enumerate}
\usepackage{enumitem}
\usepackage{bm}
\usepackage{algorithmic}
\usepackage{algorithm}
\usepackage{relsize}
\usepackage[utf8]{inputenc}
\usepackage{cleveref}
\usepackage{url}
\usepackage{lipsum}
\usepackage{fancyhdr}
\usepackage{tikz}
\usetikzlibrary{calc}
\crefname{section}{§}{§§}
\Crefname{section}{§}{§§}

\newtheorem{theorem}{Theorem}
\newtheorem{lemma}{Lemma}
\newtheorem{definition}{Definition}

\newtheorem{corollary}{Corollary}
\newtheorem{claim}{Claim}

\IEEEoverridecommandlockouts
\begin{document}

\title{Coflow Scheduling in Input-Queued Switches: Optimal Delay Scaling and Algorithms}
\author{
\IEEEauthorblockN{Qingkai Liang and Eytan Modiano}
\IEEEauthorblockA{Laboratory for Information and Decision Systems\\Massachusetts Institute of Technology, Cambridge, MA}
\thanks{This work was supported by NSF Grants CNS-1116209, CNS-1617091, and by DARPA I2O and Raytheon BBN Technologies under Contract No. HROO l l-l 5-C-0097.}
}

\maketitle

\begin{tikzpicture}[remember picture, overlay]
\node at ($(current page.north) + (-3in,-0.5in)$) {Technical Report};
\end{tikzpicture}

\begin{abstract}
A coflow is a collection of parallel flows belonging to the same job. It has the all-or-nothing property: a coflow is not complete until the completion of all its constituent flows. In this paper, we focus on optimizing \emph{coflow-level delay}, i.e., the time to complete all the flows in a coflow, in the context of an $N\times N$ input-queued switch.  In particular, we develop a throughput-optimal scheduling policy that achieves the best scaling of coflow-level delay as $N\rightarrow\infty$. We first derive lower bounds on the coflow-level delay that can be achieved by any scheduling policy. It is observed that these lower bounds critically depend on the variability of flow sizes. Then we analyze the coflow-level performance of some existing coflow-agnostic scheduling policies and show that none of them achieves provably optimal performance with respect to coflow-level delay. Finally, we propose the Coflow-Aware Batching (CAB) policy which achieves the optimal scaling of coflow-level delay under some mild assumptions.
\end{abstract}

\section{Introduction}
Modern cluster computing frameworks, such as MapReduce \cite{mapreduce} and Spark \cite{spark}, have been widely used in large-scale data processing and analytics. Despite the differences among these frameworks, they share a common feature:  the computation is divided into multiple stages and a collection of parallel data flows need to be transferred between groups of machines in successive computation stages. Often the next computation stage cannot start until the completion of all of these flow transfers. For example, during the shuffle phase in MapReduce, any reducer node cannot start the next reduce phase until it receives intermediate results from all of the mapper nodes. As a result, the response time of the entire computing job critically depends on the completion time of these intermediate flows. In some applications, these intermediate flow transfers can account for more than 50\% of job completion time \cite{orchestra}.

The recently proposed \emph{coflow} abstraction \cite{ion-coflow2} represents such a collection of parallel data flows between two successive computation stages of a job, which exposes application-level requirements to the network. It builds upon the \emph{all-or-nothing} property observed in many applications \cite{coflow-no-prior}: a coflow is not complete until the completion of all its constituent flows. As a result, one of the most important metrics in this context is \emph{coflow-level delay} (also referred to as \emph{coflow completion time} in some literature \cite{coflow-ion, coflow-no-prior, rapier}), i.e., the time to complete all of the flows in a coflow. To improve the overall response time of a job, it is crucial to schedule flow transfers in a way that the coflow-level delay can be reduced. Unfortunately, researchers have largely overlooked such application-level requirements and there has been little work on coflow-level delay optimization.

In this paper, we study coflow-level delay in the context of an $N\times N$ input-queued switch with stochastic coflow arrivals. In each slot, a random number of coflows, each of them consisting of multiple parallel flows, arrive to the input-queued switch where each input/output port can process at most one packet per slot. Such an input-queued switch model is a simple yet practical abstraction for data centers with full bisection bandwidth, where $N$ represents the number of servers. Due to the large scale of modern data centers, we are motivated to study the scaling of coflow-level delay as $N\rightarrow\infty$. In particular, we are interested in the optimal scaling of coflow-level delay, i.e., the scaling under an ``optimal" scheduling policy\footnote{An ``optimal" policy should be throughput-optimal, i.e., stabilize the system whenever the load $\rho<1$, and should achieve the minimum coflow-level delay among all throughput-optimal policies.}. As far as we know, this is the first paper to present coflow-level delay analysis in a large-scale stochastic system.

The contributions of this paper are summarized as follows.

\begin{itemize}[leftmargin=0.3cm,itemsep=1mm,topsep=0.5mm]
\item{We derive lower bounds on the expected coflow-level delay that can  be achieved by any scheduling policy in an $N\times N$ input-queued switch. These lower bounds critically depend on the variability of flow sizes. In particular, it is shown that if flow sizes are light-tailed, no scheduling policy can achieve an average coflow-level delay better than $O(\log N)$.}

\item{We analyze the coflow-level performance of several coflow-agnostic scheduling policies, where coflow-level information is not leveraged. It is shown that none of these scheduling policies achieves a \emph{provably optimal} scaling of coflow-level delay. For example, the expected coflow-level delay achieved by randomized scheduling is $O(N\log N)$ if coflow sizes are light-tailed, far above the $O(\log N)$ lower bound.}

\item{We show that $O(\log N)$ is the optimal scaling of average coflow-level delay when flow sizes are light-tailed and coflow arrivals are Poisson. This optimal scaling is achievable with our Coflow-Aware Batching (CAB) policy.}
\end{itemize}

The organization of this paper is as follows. We first review related work in Section \ref{sec:related} and introduce several mathematical tools in Section \ref{sec:pre}. The system model is introduced in Section \ref{sec:model}. In Section \ref{sec:lower}, we demonstrate fundamental lower bounds on the expected coflow-level delay that can be achieved by any scheduling policy. In Section \ref{sec:agnostic}, we analyze the coflow-level performance of some coflow-agnostic scheduling policies. In Section \ref{sec:cab}, we propose the Coflow-Aware Batching (CAB) policy and show that it achieves the optimal coflow-level delay scaling under some conditions. Finally, simulation results and conclusions are given in Sections \ref{sec:simulation} and \ref{sec:conclusion}, respectively.

\section{Related Work}\label{sec:related}
We start with a brief literature review on coflow-level optimization and delay scaling in input-queued switches.

\vspace{1mm}

\noindent \textbf{Coflow-level Optimization.} The notion of coflows was first proposed by Chowdhy and Stoica \cite{ion-coflow2} to convey job-specific requirements such as minimizing coflow-level delay or meeting some job completion deadline. Unfortunately, coflow-level optimization is often computationally intractable. For example, it was shown in \cite{coflow-ion} that minimizing the average coflow-level delay is NP-hard. As a result, many heuristic scheduling principles were developed to improve coflow-level delay. In \cite{baraat}, a decentralized coflow scheduling framework was proposed to give priority to coflows according to a variation of the FIFO principle, which performs well for light-tailed flow sizes. In \cite{coflow-ion}, the $\mathsf{Varys}$ scheme improves the performance of \cite{baraat} by leveraging more sophisticated heuristics such as ``smallest-bottleneck-first" and ``smallest-total-size-first", where global information about coflows is required. The D-CAS scheme in \cite{D-CAS} exploits a similar ``shortest-remaining-time-first" principle for coflow scheduling. The $\mathsf{Aalo}$ framework \cite{coflow-no-prior} generalizes the classic least-attained service (LAS) discipline \cite{las} to coflow scheduling; such a scheme does not require prior knowledge about coflows. Zhong \emph{et al.} \cite{approx-coflow} develop an approximation algorithm to minimize the average coflow-level delay in data centers.  Additionally, Chen \emph{et al.} \cite{rapier} jointly consider coflow routing and scheduling in data centers.  Despite these efforts towards coflow-level optimization, most prior works do not provide any analytical performance guarantee, and there is a lack of fundamental understanding of coflow-level scheduling, especially in the context of large-scale stochastic systems.

\vspace{1mm}

\noindent \textbf{Optimal Delay Scaling in Input-Queued Switches.} The optimal (\emph{packet-level}) delay scaling in input-queued switches (i.e., the delay scaling under an optimal scheduling policy) has been an important area of research for more than a decade. The randomized scheduling policy \cite{ct} (based on Birkhoff-Von Neumann decomposition) achieves an average packet delay of $O(N)$. The well-known Max-Weight Matching (MWM) \cite{nick-max-weight} policy and various approximate MWM algorithms \cite{max-weight-delay-scaling, shah-delay-max} are shown to have an average packet delay no greater than $O(N)$, although it is conjectured that this bound is not tight  for a wide range of traffic patterns \cite{max-weight-delay-scaling}. Recently, Maguluri \emph{et al.} \cite{srikant-max1, srikant-max2} show that MWM can achieve the optimal $O(1)$ packet-level delay in the \emph{heavy-traffic regime}.  Neely \emph{et al.} \cite{Modiano-batching} propose a batching scheme that achieves an average packet delay of $O(\log N)$; this is the best known result for packet-level delay scaling as $N\rightarrow\infty$ \emph{under general traffic conditions}. Zhong \emph{et al.} \cite{optimal-scaling} consider the joint scaling of queue length as $N\rightarrow\infty$ and $\rho\rightarrow 1$. They propose a policy that gives an upper bound of $O(1\slash(1-\rho)+N^2)$; this joint scaling is shown to be ``optimal" in the \emph{heavy-traffic regime} where $\rho=1-O(\frac{1}{N^2})$. However, to the best of our knowledge, the optimal scaling of packet-level delay under a general traffic condition is still an open problem in input-queued switches. By comparison, the optimal scaling of \emph{coflow-level delay}, which is an upper bound for packet-level delay, has not been studied before. In this paper, we make the first attempt in deriving the optimal coflow-level delay scaling as $N\rightarrow\infty$.
\section{Preliminaries}\label{sec:pre}
In this section, we briefly introduce some common notations and useful mathematical tools that facilitates our subsequent analysis.
\subsection{Notation}
Define $[N]=\{1,\cdots,N\}$. We reserve bold letters for matrices. For example, $\mathbf{X}=(X_{ij})_{N\times N}$ denotes an $N\times N$ matrix with the $(i,j)$-th element being $X_{ij}$. For simplicity of notation, we may drop the subscript ``$N\times N$" if the context is clear.

We use the traditional asymptotic notations. Let $f$ and $g$ be two functions defined on some subset of real numbers. Then $f(n)=O(g(n))$ if there exists some positive constant $C$ and a real number $n_0$ such that $|f(n)|\le C|g(n)|$ for all $n\ge n_0$;  similarly,  $f(n)=\Omega(g(n))$ if there exists some positive constant $C$ and a real number $n_0$ such that $f(n)\ge Cg(n)$ for all $n\ge n_0$; finally, $f(n)=\Theta(g(n))$ if $f(n)=O(g(n))$ and $f(n)=\Omega(g(n))$.
\subsection{Mathematical Tools}\label{sec:math}
\vspace{1mm}
\noindent \textbf{(1). Stochastic Dominance.} We consider the first-order stochastic dominance whose definition is as follows.
\begin{definition}[Stochastic Dominance \cite{stochastic-dominance-1}]
Consider \\two random variables $X_1$ and $X_2$. Then $X_1$ stochastically dominates $X_2$ if $\mathbb{P}[X_1\ge x]\ge \mathbb{P}[X_2\ge x]$ for all $x\in\mathbb{R}$.
\end{definition}
Intuitively, stochastic dominance defines the ``inequality relationship" between two random variables in the probabilistic sense. If $X_1$ stochastically dominates $X_2$, then $X_1$ has an equal or higher probability of taking on a large values than $X_2$. Two useful properties of stochastic dominance are as follows \cite{stochastic-dominance-1}.
\begin{itemize}
\item[(P1)] Consider two non-negative random variables $X_1$ and $X_2$. If $X_1$ stochastically dominates $X_2$, then $\mathbb{E}[X^n_1]\ge \mathbb{E}[X^n_2]$ for all $n\in\mathbb{Z}^+$.

\item[(P2)] Suppose $\{X_1,\cdots,X_N\}$ is a set of independent random variables and $\{Y_1,\cdots,Y_N\}$ is another set of independent random variables. If $X_i$ stochastically dominates $Y_i$ for all $i\in[N]$, then $\max_i X_i$ also stochastically dominates $\max_i Y_i$.
\end{itemize}
The above properties will be helpful when establishing inequalities among expectations (and higher moments) of random variables.

\vspace{1mm}

\noindent \textbf{(2). Associated Random Variables.} The association of random variables is a stronger notion of positive correlation. The formal definition is as follows.

\begin{definition}[Associated Random Variables \cite{association-random}]
A collection of random variables $X_1,\cdots,X_n$ are said to be associated if $\text{Cov}(f(\mathbf{X}),g(\mathbf{X}))\ge 0$ for all pairs of non-decreasing functions $f$ and $g$.
\end{definition}
Intuitively, if a collection of random variables are associated, they are usually positively correlated (at least independent). In other words, if $X_1$ takes on a large value, then $X_2,\cdots,X_n$ are also very likely to take on large values. The followings are some useful properties of associated random variables (see Appendix A of \cite{fork-join-mm1}).

\begin{itemize}
\item[(P1)] Independent random variables are associated (trivial case).

\item[(P2)] Non-decreasing functions of associated random variables are also associated.

\item[(P3)] If two sets of associated random variables are independent of one another, then their union is a set of associated random variables.

\item[(P4)] If $X_1,\cdots, X_n$ are associated, then $\max_i X_i$ is stochastically dominated by $\max_i X'_i$ where $X'_i$'s are \emph{independent} random variables identically distributed as $X_i$'s.
\end{itemize}
A particularly important property is (P4) which relates the maximum of a set of (possibly dependent) random variables to the maximum of a set of independent random variables.
\section{System Model}\label{sec:model}
\subsection{Network Model}
We consider an $N\times N$ input-queued switch with $N$ input ports and $N$ output ports. The system operates in slotted time, and the slot length is normalized to one unit of time. In each slot, each input can transfer at most one packet and each output can receive at most one packet (this is referred to as ``crossbar constraints"). Such an input-queued switch model is simple yet very useful in modeling many practical networked systems. For example, data centers with full bisection bandwidth can be abstracted out as a giant input-queued switch interconnecting different machines. Note that each input port may have packets destined for different output ports, which can be represented as \emph{Virtual Output Queues} (VOQ). There are a total of $N^2$ virtual output queues, indexed by $(i,j)$ for $i,j\in[N]$, where $[N]\triangleq\{1,\cdots,N\}$ and queue $(i,j)$ holds packets from  input $i$ to output $j$. Figure \ref{fig:switch} shows a $2\times 2$ input-queued switch with four virtual output queues.

The schedule of packet transmissions in slot $t$ can be represented by an $N\times N$ matrix $\mathbf{S}(t)=(S_{ij}(t))$ where $S_{ij}(t)=1$ if the connection between input $i$ and output $j$ is activated. A feasible schedule $\mathbf{S}(t)$ is one that satisfies the crossbar constraints, i.e., $\mathbf{S}(t)$ must be a binary matrix where there is at most one ``1" in each row and each column.

\begin{figure}[]
\begin{center}
\includegraphics[width=3in]{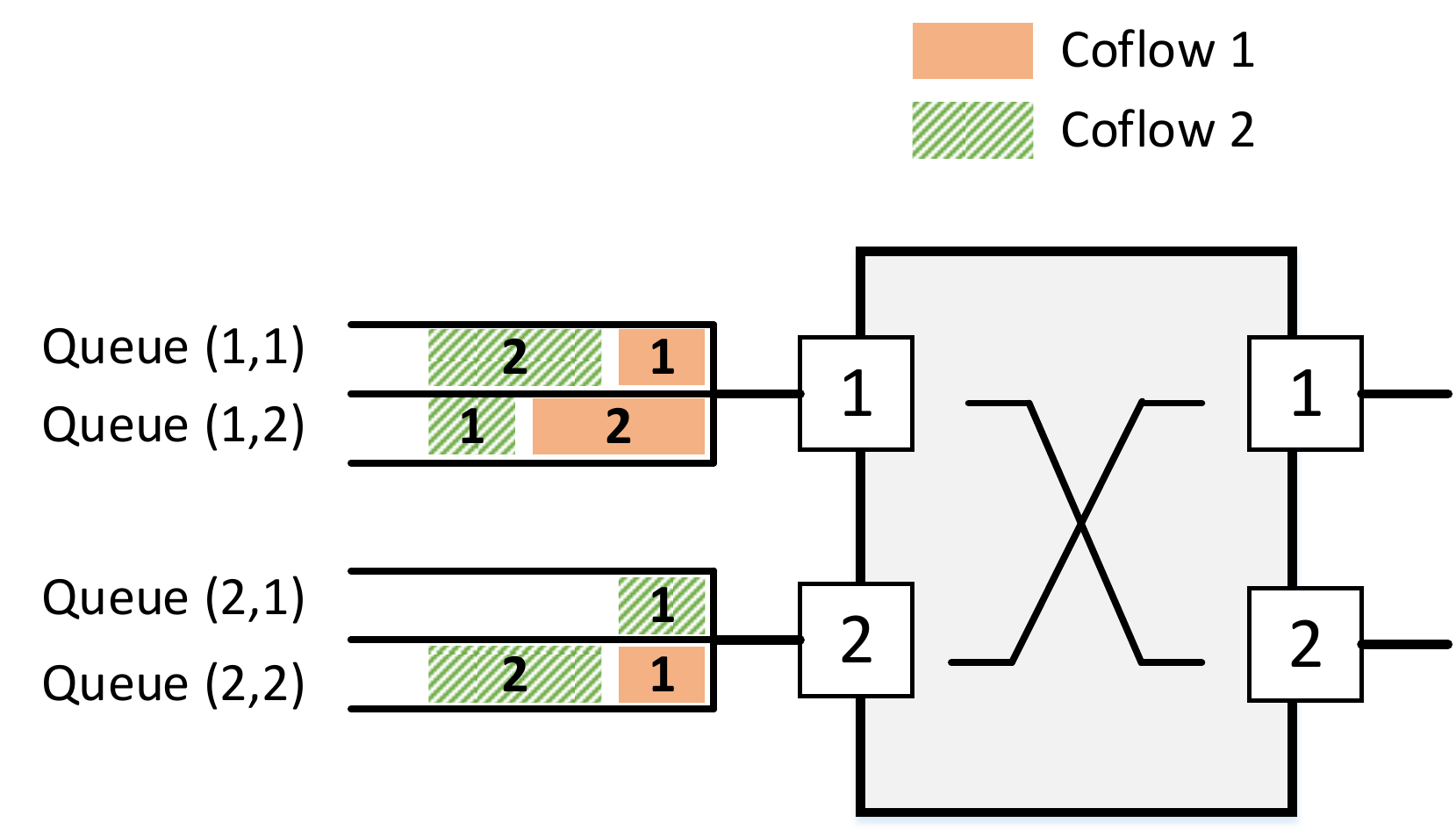}
\caption{A $2\times 2$ input-queued switch with two coflow arrivals.}
\label{fig:switch}
\end{center}
\end{figure}
\subsection{Coflow Abstraction}\label{sec:coflow-abs}
A coflow is a collection of parallel data flows belonging to the same job. It has the all-or-nothing property: a coflow is not complete until the completion of all its constituent flows. Coflows are a useful abstraction for many communication patterns in data-intensive computing applications such as MapReduce (see \cite{ion-coflow2} for  applications of coflows). Note that the traditional point-to-point communication is a special case of coflows (i.e., a coflow with a single flow).

Formally, we represent each coflow by a random traffic matrix $\mathbf{X}=(X_{ij})$ where $X_{ij}$ is the number of packets in this coflow that need to be transmitted from input $i$ to output $j$. Note that each coflow may contain many small flows from input $i$ to output $j$, that are aggregated into a single batch $X_{ij}$ for ease of exposition. In the following, $X_{ij}$ will be referred to as the ``\textbf{batch size}" or ``\textbf{flow size}" from input $i$ to output $j$. We assume that all the packets in a coflow are released simultaneously upon the arrival of this coflow. Figure \ref{fig:switch} illustrates two coflow arrivals whose traffic matrices are
$
\mathbf{X}_1=
\begin{bmatrix}
1 & 2\\
0 & 1
\end{bmatrix}
$ and $
\mathbf{X}_2=
\begin{bmatrix}
2 & 1\\
1 & 2
\end{bmatrix}.
$
Let $\mathbb{E}[X_{ij}]=\beta_{ij}$ and assume that coflows arrive to the system with rate $\lambda$. Then the arrivals of packets to queue $(i,j)$ is a \emph{batch arrival process} with rate $\lambda$ and mean batch size $\beta_{ij}$.  Also define $\beta\triangleq \max\Big(\max_i \sum_j \beta_{ij}, \max_j \sum_i \beta_{ij}\Big)$. In our analysis, we also make the following assumptions:
\begin{itemize}
\item[(1)]{The arrival times and the batch sizes of different coflows are independent.}
%
%

\item[(2)]{$X_{ij}$'s are independent random variables.}

\item[(3)]{If the $n$-th moment of $X_{ij}$ is finite, we assume $\mathbb{E}[(\sum_j X_{ij})^n]=O(1)$ for all $i\in[N]$ and $\mathbb{E}[(\sum_i X_{ij})^n]=O(1)$ for all $j\in[N]$ as $N\rightarrow\infty$.}


\item[(4)]{The sub-critical condition holds: $\rho=\lambda\beta<1$.}
\end{itemize}

In this paper, we focus on optimizing \emph{coflow-level delay}, i.e., \emph{the time between the arrival of a coflow until all the packets associated with this coflow are transmitted}. In particular, we are interested in the scaling of coflow-level delay in a large-scale system, as $N\rightarrow\infty$. Our objective is to find a scheduling policy that achieves the best dependence of coflow-level delay on $N$ while stabilizing the system whenever $\rho<1$.
\section{Lower Bounds on Coflow-level Delay}\label{sec:lower}
Before we investigate any specific scheduling policy, it is useful to study the fundamental scaling properties of coflow-level delay as $N\rightarrow\infty$. In this section, we develop lower bounds on coflow-level delay in input-queued switches. These lower bounds serve as the baselines when we evaluate the coflow-level performance of a scheduling policy. We first introduce the notion of \emph{clearance time}.
\begin{definition}[Clearance Time]
The clearance time of a coflow $\mathbf{X}=(X_{ij})$ is
\begin{equation}\label{eq:clear}
\tau(\mathbf{X})=\max \Big (\max_i \sum_{j} X_{ij}, \max_j \sum_{i} X_{ij}\Big).
\end{equation}
\end{definition}
\noindent Clearly, $\tau(\mathbf{X})$ is the maximum number of packets in a row  or a column of $\mathbf{X}$. Since each input/output port can process at most one packet per slot, the minimum time to clear all the packets in $\mathbf{X}$ must be no smaller than $\tau(\mathbf{X})$. In fact, $\mathbf{X}$ can be cleared in exactly $\tau(\mathbf{X})$ slots by using the \emph{optimal clearance algorithm} described in \cite{clearance-policy}. 
As a result,  $\tau(\mathbf{X})$ is the minimum time needed to transmit all the packets in a coflow $\mathbf{X}$. In the rest of this section, we investigate the scaling of clearance time as $N\rightarrow\infty$ and its relationship to coflow-level delay.

Depending on the distributions of  $X_{ij}$'s (the flow sizes), the scaling of clearance time exhibits different behaviors. First, we consider the general case where the flow size distribution is arbitrary (as long as $\mathbb{E}[X_{ij}]<\infty$ for all $i,j\in [N]$).
\begin{lemma}\label{thm:ct-general}
For a coflow $\mathbf{X}=(X_{ij})$ with $\mathbb{E}[X_{ij}]<\infty$, the expected clearance time is $\mathbb{E}[\tau(\mathbf{X})]=O(N)$. Moreover, there exist distributions of $X_{ij}$'s such that $\mathbb{E}[\tau(\mathbf{X})]=\Omega(N^{\frac{1}{1+\epsilon}})$ for any $\epsilon>0$.
\end{lemma}
\begin{proof}
The $O(N)$ upper bound is nearly trivial. It is clear that
\[
\begin{split}
\mathbb{E}[\tau(\mathbf{X})]&\le \mathbb{E}[\sum_{i,j}X_{ij}]\le N\beta =O(N).
\end{split}
\]

To prove the lower bound, we find some  distributions of $X_{ij}$'s such that $\mathbb{E}[\tau(\mathbf{X})]=\Omega(N^{\frac{1}{1+\epsilon}})$ for any $\epsilon>0$. Consider the scenario where $\mathbf{X}$ is a diagonal matrix: $X_{ij}=0$ with probability 1 for $i\ne j$ and $X_{ii}$ has the power law for all $i\in[N]$, i.e.,
\[
\mathbb{P}[X_{ii}\ge k]=\frac{1}{k^{1+\epsilon}},~k=1,2,\cdots,
\]
where $\epsilon>0$. Note that $\mathbb{E}[X_{ii}]=\frac{1+\epsilon}{\epsilon}$ but it can be easily scaled or shifted to have an arbitrary expectation. By simple calculation on the order statistics, we have
\[
\begin{split}
\mathbb{E}[\tau(\mathbf{X})]&\ge \mathbb{E}[\max_i \sum_j X_{ij}]\\
&=\mathbb{E}[\max_i X_{ii}]\\
&=\sum_{k=1}^{\infty} \Big[1-(1-\frac{1}{k^{1+\epsilon}})^N\Big]\\
&\ge \sum_{k=1}^{N_{\epsilon}}  \Big[1-\Big(1-\frac{1}{N^{1+\epsilon}_{\epsilon}}\Big)^N\Big]\\
& \ge N_{\epsilon}  \Big[1-\Big(1-\frac{1}{N}\Big)^N\Big],
\end{split}
\]
where $N_{\epsilon}=\lfloor N^{\frac{1}{1+\epsilon}}\rfloor$. Since $1-\Big(1-\frac{1}{N}\Big)^N \xrightarrow{N\rightarrow\infty}1-e^{-1}$, we can conclude that $N_{\epsilon}  \Big[1-\Big(1-\frac{1}{N}\Big)^N\Big]=\Theta(N_{\epsilon})=\Theta(N^{\frac{1}{1+\epsilon}})$ as $N\rightarrow\infty$. As a result, there exist (heavy-tailed) distributions of $X_{ij}$'s such that $\mathbb{E}[\tau(\mathbf{X})]=\Omega(N^{\frac{1}{1+\epsilon}})$ for any $\epsilon>0$.
\end{proof}

If $X_{ij}$ also has a finite variance, we can obtain a better scaling behavior of clearance time as $N\rightarrow\infty$. In this case, we  assume $\text{Var}(\sum_{j}X_{ij})\le \sigma^2$ for all $i\in[N]$ and $\text{Var}(\sum_{i}X_{ij})\le \sigma^2$ for all $j\in[N]$, where $\sigma^2$ is a constant independent of $N$.
\begin{lemma}\label{thm:ct-second}
For a coflow $\mathbf{X}=(X_{ij})$ with $\text{Var}(X_{ij})<\infty$ for all $i,j\in[N]$, the expected clearance time is $\mathbb{E}[\tau(\mathbf{X})]=O(\sqrt{N})$. Moreover, there exist distributions of $X_{ij}$'s such that $\mathbb{E}[\tau(\mathbf{X})]=\Omega(N^{\frac{1}{2+\epsilon}})$ for any $\epsilon>0$.
\end{lemma}
\begin{proof}
Devroye \cite{order-root} shows that if $Y_1,\cdots,Y_n$ are (possibly dependent) random variables with finite means and finite variances, then
\[
\mathbb{E}[\max_i Y_i]\le \max_i \mathbb{E}[Y_i]+\sqrt{n}\max_{i}\sqrt{\text{Var}(Y_i)}.
\]
If $\text{Var}(X_{ij})<\infty$ for all $i,j\in [N]$, it follows that
\[
\begin{split}
\mathbb{E}[\max_i \sum_{j} X_{ij}]&\le  \max_i \mathbb{E}[\sum_{j} X_{ij}]+\sqrt{N}\max_{i}\sqrt{\text{Var}(\sum_{j} X_{ij})}\\
&\le \beta+\sqrt{N}\sigma.
\end{split}
\]
Similarly, it can be shown that $\mathbb{E}[\max_j \sum_{i} X_{ij}]\le \beta+\sqrt{N}\sigma$. As a result, we have
\[
\mathbb{E}[\tau(\mathbf{X})]\le \mathbb{E}[\max_i \sum_{j} X_{ij}]+\mathbb{E}[\max_j \sum_{i} X_{ij}]\le 2\beta+2\sqrt{N}\sigma.
\]
Consequently, $\mathbb{E}[\tau(\mathbf{X})]=O(\sqrt{N})$. The $\Omega(N^{\frac{1}{2+\epsilon}})$ lower bound can be proved in a similar way to Lemma \ref{thm:ct-general} with the power being $2+\epsilon$ instead of $1+\epsilon$ such that the variance is finite.
\end{proof}

Furthermore, if $X_{ij}$'s have light-tailed\footnote{In this paper, a light-tailed distribution is the one with a finite Moment Generating Function in the neighborhood of 0. In other words, it has an exponentially decreasing tail.} distributions, the scaling of clearance time is logarithmic.
\begin{lemma}\label{thm:ct-light}
For a coflow $\mathbf{X}=(X_{ij})$ with light-tailed $X_{ij}$'s, the expected clearance time is $\mathbb{E}[\tau(\mathbf{X})]=O(\log N)$. Moreover, there exist distributions of $X_{ij}$'s such that this bound is tight.
\end{lemma}
\begin{proof}
See Appendix \ref{appendix:ct-light}.
\end{proof}

Finally, if $X_{ij}$'s are deterministic, it is clear that $\mathbb{E}[\tau(\mathbf{X})]\le\beta=\Theta(1)$.   Since clearance time is the minimum time to transmit all the packets in a coflow, it is a natural lower bound on coflow-level delay (which is the time between the arrival of a coflow until all the packets associated with this coflow are transmitted). Consequently, the above results essentially impose fundamental limits on the coflow-level delay that can be achieved by any scheduling policy.

\begin{table}[]
\centering
\caption{Lower Bounds on Coflow-level Delay}\label{tb:lower}
\small
\begin{tabular}{|c|c|c|}
\hline
Condition                    & Coflow-level Delay                                  \\ \hline
$\mathbb{E}[X_{ij}]<\infty$    & $\Omega(N^{\frac{1}{1+\epsilon}})$ for any $\epsilon>0$ \\ \hline
$\text{Var}(X_{ij})<\infty$   & $\Omega(N^{\frac{1}{2+\epsilon}})$ for any $\epsilon>0$ \\ \hline
$X_{ij}$'s are light-tailed   & $\Omega(\log N)$                                    \\ \hline
$X_{ij}$'s are deterministic   & $\Omega(1)$                                    \\ \hline
\end{tabular}
\end{table}

\begin{theorem}\label{coro:lower-bound}
The expected coflow-level delay achieved by any scheduling policy cannot be better than $O(g(N))$, where\\
(1) $g(N)=N^{\frac{1}{1+\epsilon}}$ for any $\epsilon>0$ if $\mathbb{E}[X_{ij}]<\infty$;\\
(2) $g(N)=N^{^{\frac{1}{2+\epsilon}}}$ for any $\epsilon>0$ if  $\text{Var}(X_{ij})<\infty$;\\
(3) $g(N)=\log N$ if  $X_{ij}$'s have light-tailed distributions;\\
(4) $g(N)=1$ if $X_{ij}$'s are deterministic.
\end{theorem}
\noindent The scaling properties of expected coflow-level delay are summarized in Table \ref{tb:lower}. It can be observed that the lower bound on coflow-level delay critically depends on the \emph{variability} of $X_{ij}$'s: the less $X_{ij}$'s vary, the smaller lower bound on coflow-level delay we can obtain. 

In the rest of this paper, we mainly focus on the case where $X_{ij}$'s have light-tailed distributions unless otherwise stated. The heavy-tailed case is  left for future work.
\section{Coflow-agnostic Scheduling}\label{sec:agnostic}
To gain further insights into the design of coflow-level scheduling policies, we study the performance of some \emph{coflow-agnostic} scheduling policies where coflow-level information (e.g., which packets/flows belong to the same coflow) is not leveraged. In particular, we study the coflow-level performance of two simple scheduling policies: randomized scheduling and periodic scheduling.

\vspace{1mm}

\noindent \textbf{Randomized Scheduling.} Let $M_1,\cdots,M_{N!}$ be the $N!$ perfect matchings (permutation matrices) associated with the $N\times N$ switch. With the Birkhoff-Von Neumann decomposition, we can find probabilities $\{p_1,p_2,\cdots,p_{N!}\}$ such that the matrix $(\lambda\beta_{ij}) \le \sum_{k=1}^{N!}p_kM_k$ (where $\beta_{ij}=\mathbb{E}[X_{ij}]$). Such a decomposition is always feasible since $(\lambda\beta_{ij})$ is sub-stochastic by Assumption (4) in Section \ref{sec:coflow-abs}. In each slot, the randomized policy uses matching $M_k$ as the schedule, with probability $p_k$. Under uniform traffic, a simple way to implement the randomized policy is to connect the $N$ input ports with a random permutation of the $N$ output ports. Such a policy is guaranteed to stabilize the network as long as $\rho<1$ although $\lambda$ and $(\beta_{ij})$ need to be known in advance. The detailed description of this policy can be found in \cite{ct} and it can be easily shown that the randomized policy achieves $O(N)$ average \emph{packet-level} delay \cite{Modiano-batching}.

\vspace{1mm}

\noindent \textbf{Periodic Scheduling.} This policy is similar to randomized scheduling except that the scheduling decisions are deterministic. Specifically, for some (sufficiently long) period $T$, we use matching $M_k$ for exactly $p_k T$ times every $T$ slots. Under uniform traffic, a simple way to implement periodic scheduling is to connect each input port $i$ to output port $[(i+t)\text{ mod }N]+1$ in slot $t$. This policy also achieves $O(N)$ average packet-level delay whenever $\rho<1$ \cite{Modiano-batching}.

Now we analyze the coflow-level delay achieved by the above two policies. In contrast to the simple analysis of packet-level delay, it is non-trivial to analyze the coflow-level delay achieved by these policies, due to the correlation between packets (e.g., packets belonging to the same coflow arrive simultaneously). For ease of exposition, we assume that traffic is uniform such that $\mathbb{E}[X_{ij}]=\frac{\beta}{N}$ and $\text{Var}(X_{ij})=\frac{\sigma^2}{N}$ for all $i,j\in[N]$. We also assume that coflow arrivals are Poisson. The analysis can be easily extended to the general case. 
\begin{theorem}\label{thm:random-delay}
Suppose $X_{ij}$'s have light-tailed distributions. The expected coflow-level delay achieved by the randomized or periodic scheduling policy is $O(N\log N)$ whenever $\rho<1$.
\end{theorem}
\begin{proof}
See Appendix \ref{appendix:random}.
\end{proof}
\noindent \textbf{Remark 1.} The proof to Theorem \ref{thm:random-delay} also shows that the average coflow-level delay achieved by the randomized or periodic policy is $O(\frac{N\log N}{1-\rho})$ as $\rho\rightarrow1$ and $N\rightarrow\infty$.
%


%

\vspace{1mm}

\noindent \textbf{Remark 2.} Comparing with the $O(N)$ packet-level delay, we can observe a coflow-level delay ``dilation" factor of $O(\log N)$. Intuitively, the delay dilation is due to the additional ``assembly delay":  packets processed earlier must wait for packets (in the same coflow) that are processed later. 

Finally, it is worth mentioning that the randomized or periodic scheduling policy is the simplest throughput-optimal policy in input-queued switches, but it sheds light on the non-triviality of coflow-level analysis (e.g., the correlation between packets) and the potential weakness of coflow-agnostic algorithms (as can be seen from the coflow-level delay dilation). The coflow-level analysis of more sophisticated policies, such as MaxWeight Matching (MWM) scheduling, are very challenging and left for future work. In fact, even the \emph{packet-level} delay of MWM is still an open problem \cite{srikant-max1, srikant-max2}.

In conclusion, there has been no throughput-optimal scheduling policy that achieves the provably optimal scaling with respect to coflow-level delay. In the next section, we propose a new coflow-aware scheduling policy that achieves the optimal scaling of coflow-level delay while maintaining throughput optimality and requiring no traffic statistics.

\section{Coflow-Aware Scheduling}\label{sec:cab}
In this section, we develop a coflow-aware scheduling policy that achieves $O(\log N)$ expected coflow-level delay whenever $\rho<1$ under the assumption that arrivals of coflows are Poisson and flow sizes $X_{ij}$'s are light-tailed. The policy is called the Coflow-Aware Batching (CAB) scheme. Note that in Section \ref{sec:lower}, we showed that if $X_{ij}$'s are light-tailed, no scheduling policy can achieve an expected coflow-level delay better than $O(\log N)$. As a result, the CAB policy attains the lower bound, which implies that $O(\log N)$ is optimal scaling of coflow-level delay (under Poisson arrivals and light-tailed flow sizes).
\subsection{Coflow-Aware Batching (CAB) Policy}
The basic idea of the CAB policy is to group timeslots into frames of size $T$ slots and clear coflows in batches, where one batch of coflows correspond to the collection of coflows arriving in the same frame. Coflows that are not cleared during a frame are handled separately in future frames.  By properly setting the frame size $T$, the CAB policy can achieve the desirable $O(\log N)$ average coflow-level delay. Note that Neely \emph{et al.} \cite{Modiano-batching} proposed a similar batching scheme to reduce \emph{packet-level} delay. By comparison, our CAB policy explicitly leverages coflow-level information (e.g., which packets belong to the same coflow) to reduce \emph{coflow-level} delay. More importantly, as mentioned in Section \ref{sec:agnostic},  coflow-level delay analysis is fundamentally different from packet-level analysis. The detailed description of the CAB policy is as follows.

\vspace{1mm}

\noindent \hrulefill

\textbf{Coflow-Aware Batching (CAB) Scheduling Policy}

\vspace{-1.3mm}

\noindent \hrulefill

\textbf{Setup.}
\begin{itemize}
\item{Timeslots are grouped into frames of size $T$ slots.}
\end{itemize}

\textbf{Notation.}
\begin{itemize}
\item{Denote by $\mathbf{L}(k)=\big(L_{ij}(k)\big)$ the aggregate traffic matrix of all the coflows arriving in the $k$-th frame, where $L_{ij}(k)$ is the total number of packets from input $i$ to output $j$ that arrive during the $k$-th frame.}
\end{itemize}

\textbf{Procedures.}
\begin{itemize}[leftmargin=0.5cm,itemsep=1mm,topsep=0.5mm]
\item[(1)]{In the $(k+1)$-th frame, we try to clear the coflows that arrived in the $k$-th frame, i.e., $\mathbf{L}(k)$. Let $\mathbf{B}(k+1)$ be the traffic matrix we choose to clear in the $(k+1)$-th frame (which may be less than $\mathbf{L}(k)$), and denote by $\tau$ the clearance time of $\mathbf{L}(k)$. If $\tau\le T-1$, then $\mathbf{L}(k)$ can be cleared within the first $T-1$ slots in the $(k+1)$-th frame. In this case, we just set $\mathbf{B}(k+1)=\mathbf{L}(k)$. Note that we only use the first $T-1$ slots in a frame while the remaining slot is reserved for clearing ``overflow" coflows as discussed below. If $\tau>T-1$, then \emph{overflow} occurs and only a subset of $\mathbf{L}(k)$ can be cleared in the $(k+1)$-th frame. In this case, we sequentially add coflows to $\mathbf{B}(k+1)$ in order of their arrival in the $k$-th frame until $\mathbf{B}(k+1)$ becomes maximal, i.e., adding any other coflow will make the clearance time of $\mathbf{B}(k+1)$ exceed $T-1$. If a coflow is selected to $\mathbf{B}(k+1)$, it is referred to as a \emph{\textbf{conforming coflow}} otherwise it is called a \emph{\textbf{non-conforming coflow}}.}
\item[(2)]{All the conforming coflows that arrive in the $k$-th frame are scheduled during the $(k+1)$-th frame by clearing $\mathbf{B}(k+1)$ in minimum time using an optimal clearance algorithm (e.g., see \cite{clearance-policy}).}
\item[(3)]{All the non-conforming coflows are put into a separate FIFO queue. In the last slot of each frame, this FIFO queue gets served by the switch. Note that non-conforming coflows are served one at a time, and the service time (measured in the number of \emph{frames}) of each non-conforming coflow is its clearance time. 
}
\end{itemize}
\noindent \hrulefill\\
In words, the first $T-1$ slots in a frame are used to serve conforming coflows arriving in the previous frame and the remaining slot is reserved to serve non-conforming coflows in a FIFO manner. Note that conforming coflows (that arrive in the same frame) are cleared together in a batch while non-conforming coflows are served one at a time in the separate FIFO queue. Under the CAB policy, either all the packets in a coflow are conforming or none of them are conforming.

\subsection{Performance of the CAB policy}\label{sec:cab-performance}
The following theorem shows that the CAB policy achieves $O(\log N)$ expected coflow-level delay whenever $\rho<1$ (under Poisson coflow arrivals and light-tailed flow sizes).
\begin{theorem}[Average Coflow-level Delay]\label{thm:cab-log}
Suppose coflows arrive according to a Poisson process and flow sizes are light-tailed. By selecting a proper frame size $T=O(\log N)$, the CAB policy achieves $O(\log N)$ expected coflow-level delay if $\rho<1$.
\end{theorem}
\noindent The choice of $T$ will be specified later in Section \ref{sec:proof1}. In fact, the CAB policy not only guarantees that the \emph{average} coflow-level delay is $O(\log N)$ but also ensures that the $O(\log N)$ delay is achievable for an arbitrary coflow with high probability.
\begin{corollary}[Tail Coflow-level Delay]\label{co:cab-log-worst}
By selecting a proper frame size $T=O(\log N)$, the CAB policy achieves $O(\log N)$ delay for an arbitrary coflow with probability $1-O(\frac{1}{N^2})$ whenever $\rho<1$.
\end{corollary}
\noindent In the following, we present a proof for the above results. The proof itself suggests the choice of $T$.

\subsection{Proof to Theorem \ref{thm:cab-log} and Corollary \ref{co:cab-log-worst}}\label{sec:proof1}
For simplicity, we assume that traffic is uniform such that $\mathbb{E}[X_{ij}]=\frac{\beta}{N}$ and $\text{Var}(X_{ij})=\frac{\sigma^2}{N}$ for all $i,j\in[N]$. The analysis can be easily extended to non-uniform traffic.

We discuss the expected coflow-level delay experienced by a conforming and a non-conforming coflow, respectively.
\begin{itemize}[leftmargin=0.3cm,itemsep=1mm,topsep=0.5mm]
\item{Conforming coflows are cleared within 2 frames: the frame where they arrive plus the frame where they are cleared. As a result, the coflow-level delay experienced by a conforming coflow is at most $2T$ time slots, i.e.,
\[
\text{Delay(conforming)}\le 2T.
\]}
\item{A non-conforming coflow first waits for at most T slots (the frame
where it arrives) and then waits in the separate FIFO queue. As a result, the coflow-level delay experienced by a non-conforming coflow is
\[
\text{Delay(non-conforming)}\le T+\text{Delay(FIFO~queue)}.
\]}
\end{itemize}
Let $\eta$ be the long-term fraction of non-conforming coflows. Then the average coflow-level delay of an \emph{arbitrary coflow} is
\begin{equation}\label{eq:decom}
\mathbb{E}[W]\le (1-\eta)2T+\eta[T+\text{Delay(FIFO~queue)}].
\end{equation}
In the following, we choose $T=O(\log N)$ (the specific value of $T$ will be made clear later). Under such a choice of $T$, we prove that $\eta$ is miniscule and $\text{Delay(FIFO~queue)}$ is not very large. Thus, it can be concluded that $\mathbb{E}[W]=O(\log N)$.\newpage

\vspace{1mm}

\noindent \textbf{Step 1: Determine the value of $T$.}

\vspace{1mm}

We first show that the overflow probability decreases exponentially with the frame size $T$.

\begin{lemma}\label{lm:exp-decay}
Let $P_o$ be the overflow probability in an arbitrary frame. If $\rho<1$, there exists some constant $\gamma >0$ such that
\begin{equation}\label{eq:exp-decay}
P_o\le 2N\exp(-\gamma T).
\end{equation}
\end{lemma}
\begin{proof}
Note that an overflow occurs in the $k$-th frame if the clearance time of $\mathbf{L}(k)$ is greater than $T-1$ slots, i.e., if any of the following $2N$ inequalities is violated.
\begin{equation}\label{eq:violate}
\begin{split}
&\sum_{j}L_{ij}(k)< T,~i=1,\cdots, N,\\
&\sum_{i}L_{ij}(k)< T,~j=1,\cdots, N,
\end{split}
\end{equation}
where $L_{ij}(k)$ is the number of packets that arrive to  queue $(i,j)$ during the $k$-th frame. Clearly, $\sum_{j}L_{ij}(k)$ (or $\sum_{i}L_{ij}(k)$) is the total number of packets that arrive to input $i$ (or output $j$) during the $k$-th frame, which corresponds to the number of packet arrivals during $T$ time slots in a \emph{Poisson process with batch arrivals} where the arrival rate is $\lambda$ and the mean batch size is $\beta$.

Let $Y(T)$ be the total number of packet arrivals during $T$ time slots in the above \emph{batch Poisson process}.  It is clear that
$
Y(T)=\sum_{n=1}^{N(T)}B_n.
$
Here, $N(T)$ is the number of coflow arrivals during the $T$ time slots, and has a Poisson distribution with rate $\lambda T$; $B_n$ is the number of packets brought by the $n$-th coflow to a certain input or output port, which is identically distributed as $\sum_{j}X_{ij}$ or $\sum_{i}X_{ij}$ and $\mathbb{E}[B_n]=\beta$. By Wald's equality, we have $\mathbb{E}[Y(T)]=\mathbb{E}[N(T)]\mathbb{E}[B_1]=\rho T<T$ if $\rho<1$.
Suppose the Moment Generating Function (MGF) of $B_n$ is $M_{B}(s)$, and the MGF of $Y(T)$ is $M_Y(s)$. By the property of a Poisson process with batch arrivals, we have $M_Y(s)=\exp\{\lambda T[M_B(s)-1]\}$. By the Chernoff bound, we have for any $s>0$,
\begin{equation}\label{eq:cher}
\begin{split}
\mathbb{P}[Y(T)\ge T]&\le M_Y(s)e^{-sT}\\
&=\exp\Big\{\lambda T[M_B(s)-1]-sT\Big\}\\
&=\exp\big(-f(s)T\big),
\end{split}
\end{equation}
where
\begin{equation}\label{eq:fs}
f(s)=\lambda[1-M_B(s)]+s.
\end{equation}
It is clear that $f(0)=0$ and $f'(0)=-\lambda M'_B(0)+1=1-\lambda\mathbb{E}[B]=1-\rho>0$. As a result, there exists a sufficiently small $\epsilon>0$ such that $f(\epsilon)>0$. Let $\gamma\triangleq f(\epsilon)>0$. By the same argument as in the proof to Lemma \ref{lm:asy-order-light} (see Appendix \ref{appendix:ct-light}), it can be verified that $\gamma$ is a constant independent of $N$ under our assumptions on $X_{ij}$'s. Then we have
\[
\mathbb{P}[Y(T)\ge T]\le \exp(-\gamma T).
\]
Applying the union bound, we can conclude that the overflow probability (i.e., at least one of the $2N$ inequalities in \eqref{eq:violate} is violated) is bounded by
\[
P_o\le 2N\exp(-\gamma T),
\]
which completes the proof.
\end{proof}

\vspace{1mm}

\noindent \textbf{Remark 1.} Lemma \ref{lm:exp-decay} implies that if we want to keep the overflow probability below $\delta$, we can choose
$
T\ge \frac{\log (2N\slash \delta)}{\gamma},
$
where $\gamma>0$ is some constant independent of $N$. Since $T$ is an integer, we can choose
\begin{equation}\label{eq:TT2}
T=\lceil\frac{\log (2N\slash \delta)}{\gamma}\rceil.
\end{equation}
The value of $\delta$ will be specified later such that $T=O(\log N)$ and the average coflow-level delay is also $O(\log N)$. The constant $\gamma$ can also be found in a systematic way (see Section \ref{sec:discussion} for details).

\vspace{1mm}

\noindent \textbf{Remark 2.} Lemma \ref{lm:exp-decay} holds only if $\rho<1$. If $\rho\ge 1$, $\mathbb{E}[Y(T)]=\rho T\ge T$ and the Chernoff bound \eqref{eq:cher} does not hold.

\vspace{1mm}

\noindent \textbf{Step 2: Determine the value of $\eta$.}

\vspace{1mm}

Next, we derive an upper bound for the long-term fraction of non-conforming coflows, i.e.,
\[
\eta=\lim_{K\rightarrow\infty}\frac{\sum_{k=1}^K A_{non}(k)}{\sum_{k=1}^K A(k)},
\]
where $A(k)$ is the total number of coflows that arrive during the $k$-th frame and $A_{non}(k)$ is the number of non-conforming coflows in the $k$-th frame.
We begin by  identifying a stochastic bound for the number of coflow arrivals in an \textbf{overflow} frame.
\begin{lemma}\label{lm:non-conforming-bound}
Suppose $N(T)$ is a Poisson random variable with rate $\lambda T$. Given that an overflow occurs in a frame, the number of coflows arrivals in this frame is stochastically dominated by $N(T)+T\slash\beta$ when $T$ is sufficiently large.
\end{lemma}
\begin{proof}
Let $N(T)$ be the total number of coflow arrivals in an arbitrary frame of $T$ time slots. Clearly, $N(T)$ is a Poisson random variable with rate $\lambda T$. Denote by $\mathbf{X}^{(n)}=(X_{ij}^{(n)})$ the traffic matrix of the $n$-th coflow, and let $\mathbf{L}$ be the aggregate traffic matrix of these coflows, i.e., $\mathbf{L}=\sum_{n=1}^{N(T)}\mathbf{X}^{(n)}$. It is clear that an overflow occurs if the clearance time of $\mathbf{L}$ is greater than $T-1$, i.e., $\tau(\mathbf{L})\ge T$. We first find a lower bound on the overflow probability when $T$ is sufficiently large.

\begin{claim}\label{lm:overflow-lower}
There exists some $T_0>0$ such that for any $T>T_0$ the overflow probability
\[
\mathbb{P}\Big[\tau(\mathbf{L})\ge T\Big]\ge \mathbb{P}\Big[N(T)\ge T\slash\beta\Big].
\]
\end{claim}
\begin{proof}
First notice that
\begin{equation}\label{eq:5566}
\tau(\mathbf{L})\ge \max_i \sum_j L_{ij} = \max_i\sum_{n=1}^{N(T)}\sum_j X^{(n)}_{ij},
\end{equation}
where the last equality is due to the fact that $L_{ij}=\sum_{n=1}^{N(T)}X^{(n)}_{ij}$. As a result, for any $i\in[N]$
\begin{equation}\label{eq:claim1}
\mathbb{P}[\tau(\mathbf{L})\ge T]\ge \mathbb{P}[\sum_{n=1}^{N(T)}\sum_j X^{(n)}_{ij}\ge T].
\end{equation}
By elementary probability calculation, we can derive
\[
\mathbb{E}[\sum_{n=1}^{N(T)}\sum_j X^{(n)}_{ij}]=\beta\lambda T=\rho T\]
\[
\text{Var}(\sum_{n=1}^{N(T)}\sum_j X^{(n)}_{ij})=\lambda T(\beta^2+\sigma^2),
\]
where $\sigma^2\ge 0$ is the variance of $\sum_j X^{(n)}_{ij}$. Note that when $\sigma^2=0$, then $\sum_j X^{(n)}_{ij}=\beta$ with probability 1 for any $i\in[N]$, which implies that
\[
\mathbb{P}[\tau(\mathbf{L})\ge T]=\mathbb{P}[\beta N(T)\ge T]= \mathbb{P}[N(T)\ge T\slash\beta].
\]
Therefore, we only need to consider the case where $\sigma^2>0$.

Note that $\lim_{T\rightarrow\infty}N(T)=\infty$. By Central Limit Theorem, we have
\[
\lim_{T\rightarrow\infty}\mathbb{P}[\sum_{n=1}^{N(T)}\sum_j X^{(n)}_{ij}\ge T]=1-\Phi\Big(\frac{T-\rho T}{\sqrt{\lambda T(\beta^2+\sigma^2)}}\Big)\triangleq p.
\]
At the same time, since $N(T)$ is a Poisson random variable with mean $\lambda T$, we can use normal approximation to Poisson distribution:
\[
\lim_{T\rightarrow\infty}\mathbb{P}[N(T)\ge T\slash \beta]=1-\Phi\Big(\frac{T-\rho T}{\sqrt{\lambda T \beta^2}}\Big)\triangleq p'.
\]
As a result, for any $\epsilon>0$, there exist a sufficiently large $T_0$ such that for any $T>T_0$
\begin{equation}\label{eq:claim1-2}
\begin{split}
\Big|\mathbb{P}[\sum_{n=1}^{N(T)}\sum_j X^{(n)}_{ij}\ge T]-p\Big|\le \epsilon\\
\Big|\mathbb{P}[N(T)\ge T\slash \beta]-p'\Big|\le \epsilon.
\end{split}
\end{equation}
Since $p>p'$ (note that $\sigma^2>0$), we can choose $\epsilon=\frac{p-p'}{2}>0$ and conclude that for any sufficiently large $T>T_0$
\[
\begin{split}
&\mathbb{P}[\tau(\mathbf{L})\ge T]-\mathbb{P}\Big[N(T)\ge T\slash \beta\Big]\\
\ge &\mathbb{P}\Big[\sum_{n=1}^{N(T)}\sum_j X^{(n)}_{ij}\ge T\Big]-\mathbb{P}\Big[N(T)\ge T\slash \beta\Big]\\
\ge &p-\epsilon-(p'+\epsilon)=0,
\end{split}
\]
where the first inequality is due to \eqref{eq:claim1} and the second inequality is due to \eqref{eq:claim1-2}.
\end{proof}

Next we evaluate the probability that there are at least $m$ coflow arrivals in an overflow frame.

\begin{claim}\label{lm:stochastic-domin1}
Given that an overflow occurs in a frame, the probability that there are at least $m$ coflow arrivals in this frame is upper bounded by $\mathbb{P}[N(T)\ge m|N(T)\ge T\slash\beta]$ when $T$ is sufficiently large, i.e.,
\[
\mathbb{P}\Big[N(T)\ge m\Big|\tau(\mathbf{L})\ge T\Big]\le \mathbb{P}\Big[N(T)\ge m\Big|N(T)\ge T\slash\beta\Big].
\]
\end{claim}
\begin{proof}
If $m\le T\slash \beta$, then $\mathbb{P}[N(T)\ge m|N(T)\ge T\slash\beta]=1$ and the upper bound naturally holds.

If $m>T\slash\beta$, then
\[
\small
\begin{split}
\mathbb{P}\Big[N(T)\ge m\Big|N(T)\ge \frac{T}{\beta}\Big]&=\frac{\mathbb{P}[N(T)\ge m,N(T)\ge \frac{T}{\beta}]}{\mathbb{P}[N(T)\ge \frac{T}{\beta}]}\\
&=\frac{\mathbb{P}[N(T)\ge m]}{\mathbb{P}[N(T)\ge \frac{T}{\beta}]}\\
&\ge \frac{\mathbb{P}[N(T)\ge m, \tau(\mathbf{L})\ge T]}{\mathbb{P}[\tau(\mathbf{L})\ge T]}\\
&= \mathbb{P}[N(T)\ge m|\tau(\mathbf{L})\ge T],
\end{split}
\]
where the first equality is due to the rule of conditional probability, the second equality holds because $\mathbb{P}[N(T)\ge m,N(T)\ge \frac{T}{\beta}] = \mathbb{P}[N(T)\ge m]$ when $m>T\slash\beta$, and the third inequality is due to Claim \ref{lm:overflow-lower} and the fact that $\mathbb{P}[N(T)\ge m]\ge \mathbb{P}[N(T)\ge m, \tau(\mathbf{L})\ge T]$.
\end{proof}

\begin{claim}\label{lm:poisson-dist}
If $N(T)$ is a Poisson random variable, then
$\mathbb{P}[N(T)\ge m+r|N(T)\ge r]\le \mathbb{P}[N(T)\ge m]$.
\end{claim}
\noindent This claim was proved in Appendix B of \cite{Modiano-batching}.

The above claims imply that when $T$ is sufficiently large
\[
\small
\begin{split}
\mathbb{P}\Big[N(T)\ge m\Big|\tau(\mathbf{L})\ge T\Big]&\le \mathbb{P}\Big[N(T)\ge m\Big|N(T)\ge \frac{T}{\beta}\Big]\\
&\le \mathbb{P}\Big[N(T)\ge m-\frac{T}{\beta}\Big]\\
&=\mathbb{P}\Big[N(T)+\frac{T}{\beta}\ge m\Big],
\end{split}
\]
where the first inequality is due to Claim \ref{lm:stochastic-domin1} and the second inequality is due to Claim \ref{lm:poisson-dist}. The above inequalities imply that given an overflow occurs in a frame, the number of coflow arrivals in this frame is stochastically dominated by $N(T)+T\slash \beta$ when $T$ is sufficiently large. This completes the proof of Lemma \ref{lm:non-conforming-bound}.
\end{proof}
With Lemma \ref{lm:non-conforming-bound}, we can find an upper bound for the long-term fraction of non-conforming coflows.
\begin{lemma}\label{lm:non-conforming-fraction}
If the overflow probability is $\delta$, the long-term fraction of non-conforming coflows among all coflows is upper bounded by $\eta\le \frac{2\delta}{\rho}$ when the frame size $T$ is sufficiently large.
\end{lemma}
\begin{proof}
According to Lemma \ref{lm:non-conforming-bound}, the expected number of non-conforming coflows in an overflow frame is at most $\lambda T+T\slash\beta$ (note that this is a loose bound since we treat all the coflows in an overflow frame as non-conforming coflows). As a result, the long-term fraction of non-conforming coflows is
\[
\begin{split}
\eta=\lim_{K\rightarrow\infty}\frac{\sum_{k=1}^K A_{non}(k)}{\sum_{k=1}^K A(k)}&\le \lim_{K\rightarrow\infty}\frac{\delta K (\lambda T+T\slash\beta)}{K\lambda T}\\
&=\frac{\delta (\lambda T+T\slash\beta)}{\lambda T}.
\end{split}
\]
Note that the inequality holds because the number of overflow frames is $\delta K$ as $K\rightarrow\infty$ and the average number of non-conforming coflows in each overflow frame is at most $\lambda T+T\slash\beta$. Noticing that $T\ge 1$ and $\rho=\lambda\beta<1$, we have
\[
\eta\le \frac{\delta(\lambda T+T\slash\beta)}{\lambda T}=\delta(1+\frac{T}{\rho T})\le \delta(\frac{1}{\rho}+\frac{1}{\rho})=\frac{2\delta}{\rho},
\]
which completes the proof.
\end{proof}

\vspace{1mm}

\noindent \textbf{Step 3: Determine delay in the separate FIFO queue.}

\vspace{1mm}

The third step is to find an upper bound for the average delay experienced by non-conforming coflows in the separate FIFO queue. Note that Lemma \ref{lm:exp-decay} shows that whenever $\rho<1$, the overflow probability $\delta$ can be made arbitrarily small by setting the frame size $T$ as in \eqref{eq:TT2}. In particular, we can choose the frame size to be $T=\lceil\frac{\log (2N\slash \delta)}{\gamma}\rceil=O(\log N^3)=O(\log N)$ to achieve the overflow probability $\delta=O(\frac{1}{N^2})$. Under such a choice of $T$, we can prove the following lemma.

\begin{lemma}\label{lm:delay-fifo}
Under a proper choice of $T$, the FIFO queue holding non-conforming coflows is stable whenever $\rho<1$ and the average delay experienced by non-conforming coflows in the FIFO queue is $O(NT^3)$.
\end{lemma}
\begin{proof}
Note that non-conforming coflows are placed in a discrete-time FIFO queue and are served one at a time. At the end of each frame, a batch of non-conforming coflows arrive to this queue, with probability $\delta$.  Since non-conforming coflows arrive to the FIFO queue in batches, the arrivals of non-conforming coflows are dependent. To circumvent this dependence, we notice that the arrivals of different batches of non-conforming coflows are independent: in each frame, there is a batch arrival with probability $\delta$, independent of any other frames. As a result, we overestimate the delay of any \emph{individual} non-conforming coflow by the delay experienced by the \emph{entire batch} of non-conforming coflows (i.e., the time between the arrival of the entire batch and the completion of all the non-conforming coflows in this batch) plus $T$ slots (the size of the frame in which the non-conforming coflow arrive).

As an overestimate, all the coflows that arrive in an overflow frame are treated as non-conforming coflows. Let $M$ be the total number of non-conforming coflows in an overflow frame, and denote by $\mathbf{\tilde{X}}^{(n)}=(\tilde{X}_{ij}^{(n)})$ the traffic of the $n$-th non-conforming coflow in this overflow frame. Let $\mathbf{\tilde{L}}=\sum_{n=1}^{M}\mathbf{\tilde{X}}^{(n)}$ be the aggregate traffic matrix associated with these non-conforming coflows. It follows that the service time for the entire batch of non-conforming coflows is $U = \tau(\mathbf{\tilde{L}})$ (measured in the number of frames since only one slot per frame is used to serve non-conforming coflows). Clearly, the service times for different batches of non-conforming coflows are independent. As a result, if the entire batch of non-conforming coflows is treated as a ``customer", the FIFO queue is a discrete-time GI/GI/1 queue with Bernoulli arrivals (of rate $\delta$ per frame) and general service time $U$. The average waiting time (measured in the number of frames)  in such a system can be exactly characterized \cite{discrete-queue}:
\begin{equation}\label{eq:gg1}
\text{Delay(FIFO queue)} = \frac{\delta\mathbb{E}[U^2]-\delta\mathbb{E}[U]}{2(1-\delta\mathbb{E}[U])}+\mathbb{E}[U].
\end{equation}
Now we evaluate $\mathbb{E}[U]$ and $\mathbb{E}[U^2]$.
Without loss of generality, we assume $\max_i \sum_j \tilde{L}_{ij}\ge \max_j \sum_i \tilde{L}_{ij}$ so that $U=\tau(\mathbf{\tilde{L}})=\max_i \sum_j \tilde{L}_{ij}$ (this simplification does not influence the scaling as $N\rightarrow\infty$).
Note that $\tilde{L}_{ij}= \sum_{n=1}^{M}\tilde{X}_{ij}^{(n)}$. Thus, we have
\[
U = \tau(\mathbf{\tilde{L}})=\max_i \sum_{n=1}^{M} \sum_j \tilde{X}_{ij}^{(n)}.
\]
By Lemma \ref{lm:non-conforming-bound}, $M$ is stochastically dominated by $N(T)+T\slash\beta$ where $N(T)$ is a Poisson random variable with rate $\lambda T$. By Lemma \ref{lm:non-conforming-batch} (see Appendix \ref{app:overflow-size}), $\sum_j \tilde{X}_{ij}^{(n)}$ is stochastically dominated by $\sum_j X_{ij}+T$ where $(X_{ij})$ is the traffic of an arbitrary coflow in an arbitrary frame. As a result, we have
\begin{equation}\label{eq:u-ex}
\begin{split}
\mathbb{E}[U]&\le \sum_i \mathbb{E}[M]\mathbb{E}\Big[\sum_j \tilde{X}_{ij}^{(n)}\Big]\\
&\le N (\lambda T+T\slash \beta)(\beta+T)=\Theta(NT^2),
\end{split}
\end{equation}
where the second inequality is due to (P1) of stochastic dominance.
At the same time, we have
\begin{equation}\label{eq:u}
\small
\begin{split}
\text{Var}(U)&\le \sum_i \text{Var}(\sum_{n=1}^{M} \sum_j \tilde{X}_{ij}^{(n)})\\
&= \sum_i \Bigg[\mathbb{E}[M]\text{Var}\Big(\sum_j \tilde{X}_{ij}^{(n)}\Big)+\mathbb{E}^2\Big[\sum_j \tilde{X}_{ij}^{(n)}\Big]\text{Var}(M)\Bigg]\\
& \le \sum_i \Bigg[\mathbb{E}[M]\mathbb{E}\big[(\sum_j \tilde{X}_{ij}^{(n)})^2\big]+\mathbb{E}^2\Big[\sum_j \tilde{X}_{ij}^{(n)}\Big]\mathbb{E}[M^2]\Bigg].
\end{split}
\end{equation}
By (P1) of stochastic dominance, we have
\begin{equation}\label{eq:u2}
\small
\begin{split}
&\mathbb{E}\Big[(\sum_j \tilde{X}_{ij}^{(n)})^2\Big]\le \mathbb{E}\Big[(\sum_j X_{ij}+T)^2\Big] = \sigma^2 + (\beta+T)^2,\\
&\mathbb{E}[M^2]\le \mathbb{E}[(N(T)+T\slash\beta)^2] = \lambda T + (\lambda T+T\slash \beta)^2.
\end{split}
\end{equation}
Taking \eqref{eq:u2} into \eqref{eq:u}, we have
\[
\small
\begin{split}
\text{Var}(U) &\le N\Bigg\{(\lambda T+\frac{T}{\beta})\sigma^2+(\beta+T)^2\Big[2\lambda T+\frac{T}{\beta}+(\lambda T+\frac{T}{\beta})^2\Big]\Bigg\}\\
& = \Theta(NT^4),
\end{split}
\]
which implies that
\begin{equation}\label{eq:u-va}
\mathbb{E}[U^2] = \text{Var}(U)+\mathbb{E}^2[U] = \Theta(N^2T^4).
\end{equation}
Taking \eqref{eq:u-ex} and \eqref{eq:u-va} into \eqref{eq:gg1}, we have
\[
\text{Delay(FIFO queue)}\le O(T^4)+O(NT^2)=O(NT^2).
\]
Note that the above delay is measured in the number of \emph{frames}, which implies that the expected delay  (measured in the number of \emph{timeslots}) experienced by a non-conforming coflow in the FIFO queue is $O(NT^3)$.

Finally, it is worth mentioning that the GI/GI/1 queue is stable whenever $\delta\mathbb{E}[U]<1$. Due to \eqref{eq:u-ex}, this is true whenever the following condition is satisfied:
\begin{equation}\label{eq:stable}
\delta N (\lambda T+T\slash \beta)(\beta+T)<1.
\end{equation}
Since $\delta=O(\frac{1}{N^2})$ and $T=O(\log N)$, the left-hand side of \eqref{eq:stable} can be made arbitrarily small (i.e., the queue is stable) for any $N\ge 1$ under a suitably small $\delta$. We will further discuss the exact value of $\delta$ in Section \ref{sec:discussion}.
\end{proof}

\noindent \textbf{Remark.} If $\rho\ge 1$, then Lemma \ref{lm:exp-decay} does not hold and the overflow probability $\delta$ cannot be made arbitrarily small regardless of the choice of $T$. In this case, the FIFO queue holding non-conforming coflows is unstable.

\vspace{1mm}

\noindent \textbf{Step 4: Putting it all together.}

\vspace{1mm}

Finally, we can evaluate the average coflow-level delay experience by both conforming and non-conforming coflows. By Lemma \ref{lm:non-conforming-fraction}, the fraction of non-conforming coflows is at most $\eta\le \frac{2\delta}{\rho}$. If $\rho<1$, Lemma \ref{lm:exp-decay} shows that we can choose the frame size to be $T=\lceil\frac{\log (2N\slash \delta)}{\gamma}\rceil=O(\log N^3)=O(\log N)$ to achieve the overflow probability $\delta=O(\frac{1}{N^2})$. Under such a choice of $T$, Lemma \ref{lm:delay-fifo} shows that $\text{Delay(FIFO queue)}=O(NT^3)$. Taking the values of $T$, $\eta$ and $\text{Delay(FIFO queue)}$ into \eqref{eq:decom}, we can conclude that the average coflow-level delay under the CAB policy is
\begin{equation}\label{eq:delay}
\begin{split}
\mathbb{E}[W]&\le (1-\frac{2\delta}{\rho})2T+\frac{2\delta}{\rho}[T+\text{Delay(FIFO queue)}]\\
& \le 2T+\frac{2\delta }{\rho}\text{Delay(FIFO queue)}\\
&= O(T)+O(\delta NT^3)\\
&=O(T)\\
&=O(\log N).
\end{split}
\end{equation}
This completes the proof to Theorem \ref{thm:cab-log}.
\subsection{Proof to Corollary \ref{co:cab-log-worst}}
The proof to Theorem \ref{thm:cab-log} shows that the coflow-level delay experienced by any conforming coflow is no greater than $2T=O(\log N)$ and the fraction of conforming coflows is more than $1-\frac{2\delta}{\rho}=1-O(\frac{1}{N^2})$. As a result, the CAB policy also ensures that the $O(\log N)$ delay is achievable for an arbitrary coflow with high probability $1-O(\frac{1}{N^2})$.

\subsection{Heuristic Improvement}\label{sec:heu}
In this section, we propose several heuristics that improve the practical performance of the CAB policy up to some constant factor (as compared to $N$). 

\vspace{1mm}

\noindent \textbf{Shortest-Clearance-Time-First (SCTF) Rule.} This rule simply means that we should first clear the coflow with the smallest clearance time. It is inspired by the optimality of the Shortest-Processing-Time rule in traditional machine scheduling literature. The SCTF rule can be leveraged when we clear conforming coflows. We first order these conforming coflows according to their clearance time. In a certain slot, suppose that queue $(i,j)$ gets scheduled. Instead of transmitting a packet in queue $(i,j)$ according to FIFO, we select to transmit a packet of the coflow  with the shortest clearance time that also has a remaining packet in queue $(i,j)$.

\vspace{1mm}

\noindent \textbf{Dynamic Frame Sizing.} This heuristic was suggested in \cite{{Modiano-batching}}. In each frame, if all the conforming coflows from the previous frame have been cleared and there are no non-conforming coflows in the system, the system starts a new frame immediately (rather than being idle for the remainder of the frame).

\subsection{Discussions}\label{sec:discussion}
\noindent \textbf{Robustness to assumptions.} Note that the Poisson assumption is not essential in the proof; a similar proof can be constructed for any arrival process such that the number of arrivals during $T$ slots has a light-tailed distribution (referred to as a light-tailed arrival process). On the other hand, if the batch size or the arrival process is not light-tailed, then the overflow probability no longer decreases exponentially with the frame size $T$ and the logarithmic bound does not hold.

\vspace{2mm}

\noindent \textbf{Joint scaling as $\rho\rightarrow 1$ and $N\rightarrow\infty$.} It is also interesting to study the joint scaling of the coflow-level delay achieved by the CAB policy as $N\rightarrow\infty$ and $\rho\rightarrow 1$. Since the majority of coflows experience a delay no greater than $O(T)$, we focus on the scaling of $T$ as $\rho\rightarrow 1$. The second-order Taylor series of $f(s)$ (see equation \eqref{eq:fs}) around $s=0$ is
\[
f(s)=(1-\rho)s-\lambda(\sigma^2+\beta^2)s^2\slash 2+O(s^3).
\]
When $\rho\rightarrow 1$, we can set $s=s^*=\frac{1-\rho}{\lambda(\sigma^2+\beta^2)}\rightarrow 0$ so that
\[
f(s^*)=\frac{(1-\rho)^2}{2\lambda(\sigma^2+\beta^2)}+O((1-\rho)^3)>0.
\]
Define $\gamma \triangleq f(s^*)=\Theta((1-\rho)^2)$. It follows from \eqref{eq:TT2} that
\[
T=\Big\lceil\frac{\log (2N\slash \delta)}{\gamma}\Big\rceil=O\Big(\frac{\log (2N\slash \delta)}{(1-\rho)^2}\Big)=O\Big(\frac{\log N}{(1-\rho)^2}\Big)
\]
as $\rho\rightarrow 1$ and $N\rightarrow\infty$. Compared with the $O(\frac{N\log N}{1-\rho})$ scaling under randomized or periodic scheduling, the CAB policy has a much better dependence on $N$ but becomes  more sensitive to $\rho$.

\vspace{2mm}

\noindent \textbf{Computational Complexity.} The computational complexity of the CAB policy can be analyzed in a similar way to the original batching policy \cite{Modiano-batching}. The computational complexity is $O(N^{1.5}\log N)$ per slot.

\vspace{2mm}

\noindent \textbf{Choosing Parameters.} In the CAB policy, the frame size is set to be $T=\lceil\frac{\log (2N\slash \delta)}{\gamma}\rceil$. As a result, we need to choose the parameters $\delta$ and $\gamma$.

We first fix $\gamma$ and discuss how to determine the value of $\delta$ (the overflow probability). The requirements on $\delta$ are:
\begin{itemize}[itemsep=-1mm,topsep=1mm]
\item[(1)]{The Left-Hand Side (LHS) of \eqref{eq:stable} needs to be made below 1 such that the system is stable.}
\item[(2)]{$\delta=O(\frac{1}{N^2})$ such that the $O(\log N)$ delay is achievable.}
\end{itemize}
Note that when $N\beta\ge 1$ (which is true for relatively large $N$), the LHS of \eqref{eq:stable} can be upper bounded by
\[
\small
LHS\le \delta N (\lambda T+T\slash \beta)(\beta+N\beta T)=\delta NT (\rho+1)(1+NT).
\]
Then we can obtain  $\delta$ by solving the following system of equations:
\begin{equation}\label{eq:para}
\small
\begin{cases}
\delta NT (\rho+1)(1+NT)=\frac{1}{2}\cr
T=\lceil\frac{\log (2N\slash \delta)}{\gamma}\rceil
\end{cases}
\end{equation}
Clearly, the solution to \eqref{eq:para} finds $\delta = O(\frac{1}{N^2 T^2})\le O(\frac{1}{N^2})$ such that the second requirement is met; the first requirement is met due to the fact that $\delta NT (\rho+1)(1+NT)$ is greater than the LHS of \eqref{eq:stable}. Note that the system utilization $\rho$ can be measured, so equations \eqref{eq:para} can solved iteratively for a given $\gamma$.

Now we discuss how to estimate the value of $\gamma$. To obtain a smaller $T$, it is desirable to have a larger $\gamma$. If we know the coflow arrival rate and the distribution of batch sizes, we can compute $\gamma$ by maximizing $f(s)$ shown in \eqref{eq:fs}, i.e.,
\begin{equation}\label{eq:gamma}
\gamma = \max_{s\ge 0}\lambda[1-M_{B}(s)]+s.
\end{equation}
If the system has no information about $\lambda$ or batch size distributions, $\gamma$ can be empirically tuned. We can first pick some (small) arbitrary value of $\gamma$ and obtain $\delta$ and $T$ by solving equations \eqref{eq:para}. The system then measures the actual overflow probability $\delta'$ under such a frame size. If $\delta'>\delta$, the value of $\gamma$ is reduced by some step size to increase the frame size $T$; otherwise the value of $\gamma$ is increased by some step size. The above procedure proceeds until $\delta'\approx\delta$, and a good value of $\gamma$ is found.

\section{Simulation Results}\label{sec:simulation}
In this section, we numerically evaluate the coflow-level performance of the  CAB policy.
\subsection{Simulation Setup}
In our simulations, coflows arrive to the system according to a Poisson process with rate $\lambda = 0.3$ (per slot). The batch sizes ($X_{ij}$'s) follow a geometric distribution with mean $\frac{\beta}{N}$ where $\beta=2.5$ (measured in the number of packets). Hence, the offered load is $\rho=0.75$. The simulation is run for a sufficiently long time ($10^6$ slots) such that the steady state is reached. The parameter $\gamma$ is obtained by solving equation \eqref{eq:gamma} offline; the parameters $\delta$ and $T$ are obtained by solving equation \eqref{eq:para}.
\subsection{Scaling with $\mathlarger{\mathlarger{N\rightarrow\infty}}$}
First, we evaluate the scaling of coflow-level delay as $N\rightarrow\infty$. The following schemes are compared:
\begin{itemize}[itemsep=1mm]
\item[(1)] CAB policy. Note that two heuristics mentioned in Section \ref{sec:heu} are leveraged.

\item[(2)] Randomized scheduling as described in Section \ref{sec:agnostic}.

\item[(3)] Max-Weight Matching (MWM) scheduling: the schedule in slot $t$ is the maximum weight matching of $\mathbf{Q}(t)$ where $\mathbf{Q}(t)=(Q_{ij}(t))$ is the queue length matrix in slot $t$.

\end{itemize}

\begin{figure}[]
\begin{center}
\includegraphics[width=2.5in]{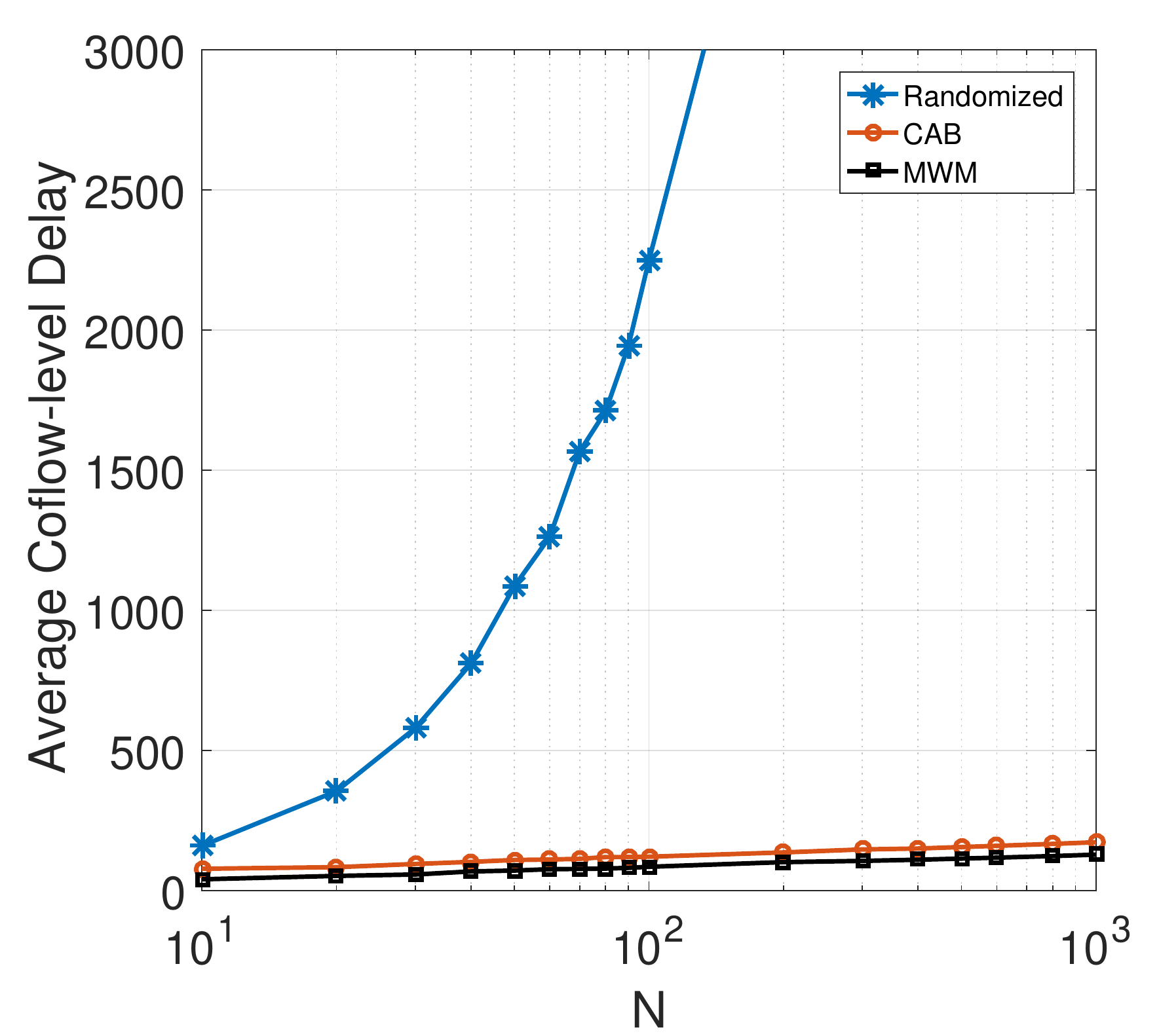}
\caption{Average coflow-level delay under different scheduling policies. Note that the horizontal axis is in the log scale.}
\label{fig:basic} \vspace{-3mm}
\end{center}
\end{figure}

\begin{figure}[]
\begin{center}
\includegraphics[width=2.8in]{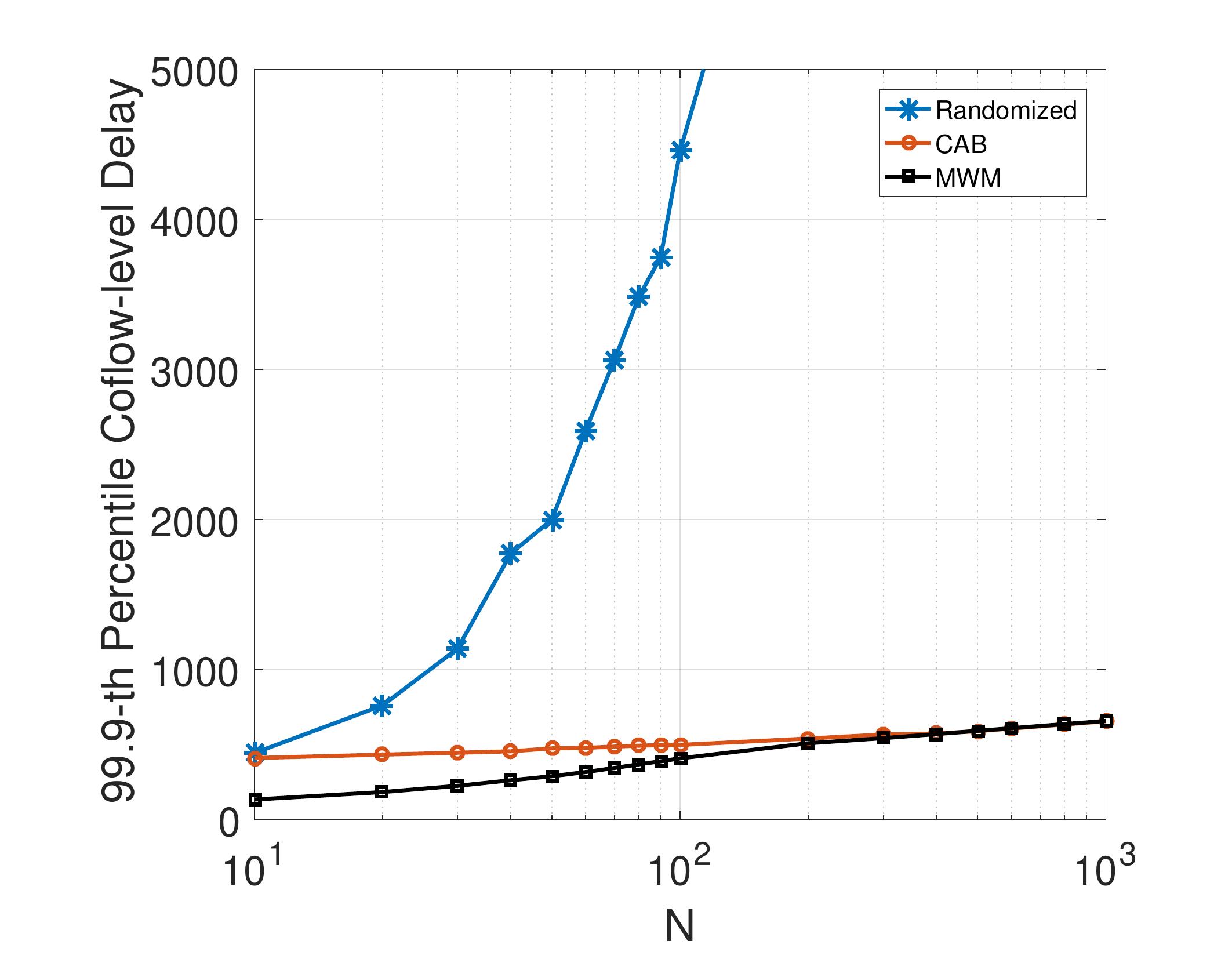}
\caption{99.9-th percentile coflow-level delay under different scheduling policies. Note that the horizontal axis is in the log scale.}
\label{fig:tail}\vspace{-3mm}
\end{center}
\end{figure}

Figure \ref{fig:basic} shows the comparison of these schemes with respect to the average coflow-level delay, where the horizontal axis is on a logarithmic scale. As the theoretical bound suggests, the CAB policy achieves the logarithmic scaling as $N\rightarrow\infty$ (i.e., a straight line in the figure). By comparison, the average coflow-level delay achieved by the randomized scheme grows much faster with $N$. Moreover, it can be observed that the CAB policy outperforms the randomized scheme even for very small $N$ (e.g.,  $N=40$). 

Another interesting observation is that the MWM policy has an exceptional coflow-level performance. It is observed that the MWM policy empirically achieves the optimal logarithmic coflow-level delay scaling as $N\rightarrow\infty$. The MWM policy also slightly outperforms the CAB policy by some constant factor. Unfortunately, the coflow-level delay analysis of the MWM policy is very challenging and left for future work.

In the above, the \emph{average} coflow-level delay is evaluated but in many cases we are also interested in \emph{tail latency}. Note that the CAB policy guarantees that the $O(\log N)$ coflow-level delay is achievable with high probability $1-O(1\slash N^2)$ (see Corollary \ref{co:cab-log-worst}), so the coflow-level delay tail under the CAB policy also scales as $O(\log N)$ when $N$ is relatively large. This is illustrated in Figure \ref{fig:tail}, where the 99.9-percentile coflow-level delay is evaluated. As expected, the CAB policy achieves $O(\log N)$ scaling for the coflow-level delay tail, significantly outperforming the randomized scheme. The MWM algorithm has a similar delay tail as the CAB policy when $N$ is relatively large.
\subsection{Coflow-level Delay Dilation}
Next, we compare the coflow-level delay with the packet-level delay under the randomized policy and the CAB policy. In particular, we are interested in the \emph{coflow-level delay dilation factor} which is the ratio between the average coflow-level delay and the average packet-level delay. As is illustrated in Figure \ref{fig:dilation}, the randomized policy has a coflow-level delay dilation factor of $O(\log N)$; this observation empirically validates the tightness of the $O(N\log N)$ bound shown in Theorem \ref{thm:random-delay} (note that the average packet delay achieved by the randomized policy is exactly $\Theta(N)$ under our simulation environment). By comparison, the delay dilation factor for the CAB policy remains at a constant level as $N\rightarrow\infty$, which shows the benefits of ``coflow-awareness".

\begin{figure}[]
\begin{center}
\includegraphics[width=2.5in]{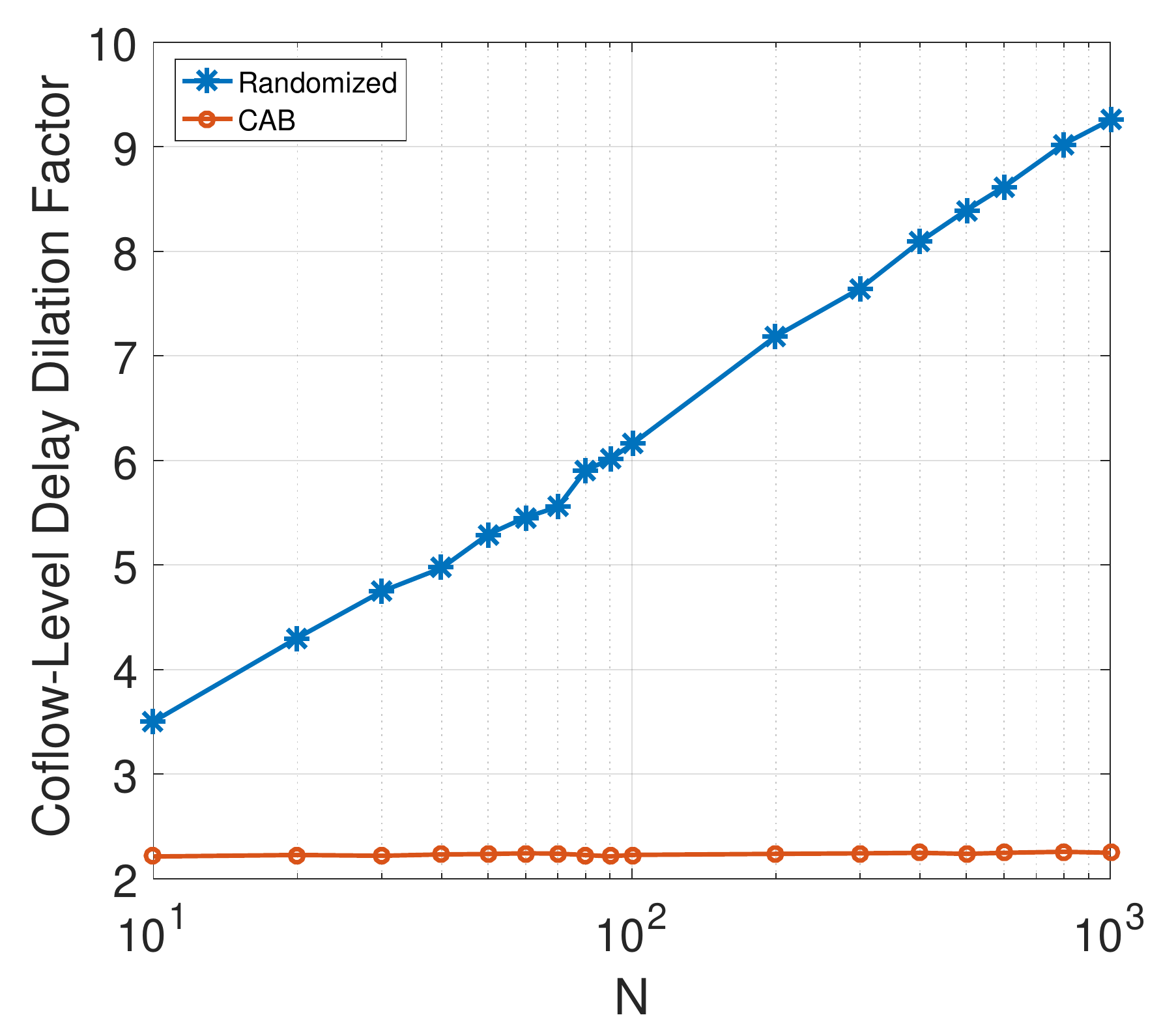}
\caption{Coflow-level delay dilation factor (the ratio between average coflow-level delay and average packet-level delay) under different scheduling policies. Note that the horizontal axis is in the log scale.}
\label{fig:dilation}
\end{center}
\end{figure}
\subsection{Scaling with $\mathlarger{\mathlarger{\rho\rightarrow 1}}$}
Finally, we numerically study the sensitivity of the coflow-level performance under different scheduling policies as the offered load $\rho\rightarrow 1$. This is shown in Figure \ref{fig:rho}. Clearly, the CAB policy is more sensitive to the offered load $\rho$ than the randomized policy. In the heavy-traffic regime, the randomized policy even outperforms the CAB policy. Indeed, the average coflow-level delay achieved by the randomized policy grows as $O(\frac{1}{1-\rho})$ as $\rho\rightarrow 1$ (see Remark 1 below Theorem \ref{thm:random-delay}). By comparison, the CAB policy achieves $O(\frac{1}{(1-\rho)^2})$ average coflow-level delay as $\rho\rightarrow 1$ (see Section \ref{sec:discussion}). As a result, the price for the better scaling with $N$ is the worse dependence on $\rho$.
\begin{figure}[]
\begin{center}
\includegraphics[width=2.5in]{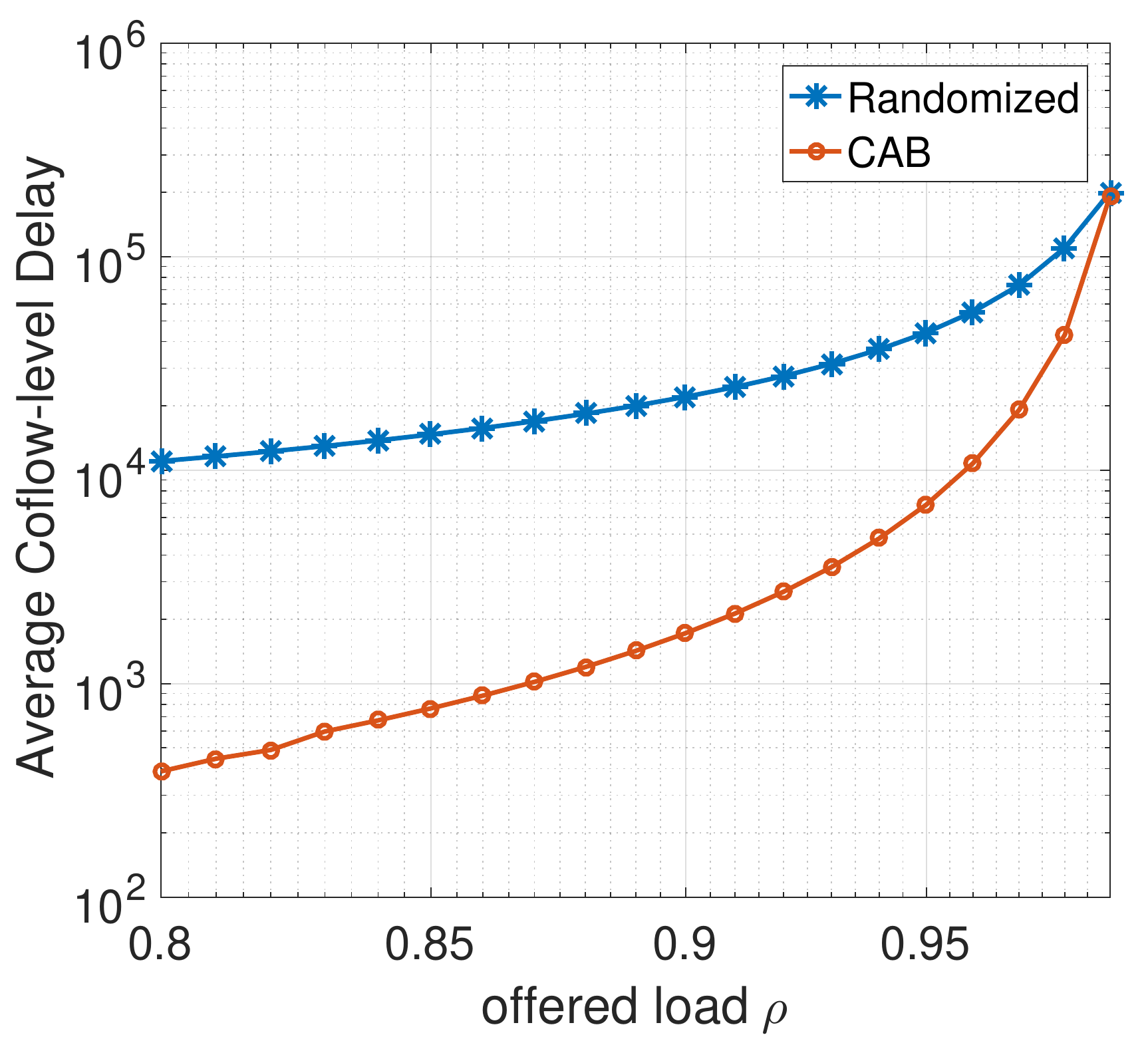}
\caption{Scaling of average coflow-level delay as $\rho\rightarrow 1$ ($N=200$). Note that the vertical axis is in the log scale.}
\label{fig:rho}
\end{center}
\end{figure}
\section{Conclusion and Future Work}\label{sec:conclusion}
In this paper, we investigate the optimal scaling of coflow-level delay in an $N\times N$ input-queued switch as $N\rightarrow\infty$. We develop lower bounds on the coflow-level delay that can be achieved by any scheduling policy. In particular, when flow sizes have light-tailed distributions, the lower bound $O(\log N)$ can be attained by the proposed Coflow-Aware Batching (CAB) policy. Thus, the optimal scaling of coflow-level delay is $O(\log N)$ under light-tailed flow sizes.

Future work includes the design of a throughput-optimal scheduling policy that achieves the best scaling of coflow-level delay under a general coflow arrival process and general flow size distributions. Variations of the Maximum Weight Matching (MWM) algorithm (that leverage coflow-level information) may be a promising direction to investigate since our simulation results show that the MWM algorithm has exceptional coflow-level performance. However, the coflow-level performance analysis of the MWM algorithm may be very challenging. Another interesting direction is to consider \emph{correlated} flow sizes and investigate how the correlation influences the scaling properties. Finally, it is also worth studying the case without prior knowledge on coflows such as coflow sizes and release times of flows (currently we assume that all the flows in a coflow are released simultaneously and coflow sizes are known).

\appendix
\section{Proofs}
\subsection{Proof to Lemma \ref{thm:ct-light}}\label{appendix:ct-light}
We first introduce a few technical lemmas. The first lemma is regarding the asymptotic bound of order statistics \cite{asy-order}.
\begin{lemma}\label{lm:asy-order}
Let $Y_1,\cdots,Y_N$ be i.i.d. $\mathbb{R}^+$-valued random variables whose common CDF $G(y)$ satisfies the following conditions:
\[
G(y) < 1,~\text{ for all } y\ge 0,
\]
and
\[
\lim_{y\rightarrow \infty}\frac{1-G(cy)}{1-G(y)}=0,~\text{ for all }c>1.
\]
Under these conditions, the asymptotics
\[
\mathbb{E}[\max_i Y_i] = m_N (1+O(1))
\]
holds true as $N\rightarrow \infty$,
where
\[
m_N = \inf\Big\{y\ge 0:1-G(y)\le \frac{1}{N}\Big\}.
\]
\end{lemma}
\noindent The second lemma is an application of Lemma \ref{lm:asy-order} to light-tailed random variables.

\begin{lemma}\label{lm:asy-order-light}
Suppose $Y_1,\cdots,Y_N$ are independent \textbf{light-tailed} random variables with $\mathbb{E}[Y_i^n]=O(1)$ as $N\rightarrow\infty$ for all $n\in \mathbb{Z}^+$. Then $\mathbb{E}[\max_i Y_i] = O(\log N)$ as $N\rightarrow\infty$. 
\end{lemma}
\begin{proof}
Suppose the CDF of $Y_i$ is $G_i(y)$, and let $\mathbb{E}[Y_i]=\mu_i$. Also let $M_{Y_i}(s)$ be the Moment Generating Function of $Y_i$. Since $Y_i$ is light-tailed, there exists some $s^*>0$ such that $M_{Y_i}(s)<\infty$ for all $0\le s\le s^*$. By Chernoff bound, for any $0\le s\le s^*$ and $y>\mu_i$
\[
1-G_i(y) \le M_{Y_i}(s)e^{-sy} \triangleq e^{-yf_i(s)},
\]
where
\[
f_i(s) = s-\frac{\log M_{Y_i}(s)}{y}.
\]
It is clear that $f_i(0)=0$ and $f_i'(0) = 1-\frac{M_{Y_i}'(0)}{yM_{Y_i}(0)}=1-\frac{\mu_i}{y}>0$ for any $y>\mu_i$. Due to the continuity of $f_i(s)$ around $s=0$, there exists a sufficiently small $\epsilon_i>0$ such that $f_i(\epsilon_i)>0$. Define $\gamma_i = f_i(\epsilon_i)>0$. We have
\[
1-G_i(y)\le e^{-\gamma_iy},~\text{for all } y> \mu_i.
\]
Let $\gamma = \min_i \gamma_i$. Then it follows that for all $i\in[N]$
\[
1-G_i(y)\le e^{-\gamma y},~\text{for all } y> \mu_i.
\]
Consider a sequence of i.i.d. random variables $Y'_1,\cdots,Y'_N$ with shifted exponential distribution $F(y)= 1- e^{-\gamma (y-\mu)}$ for $y>\mu$, where $\mu=\max_i\mu_i$. It is clear that $Y_i$ is stochastically dominated by $Y'_i$ for all $i\in[N]$. By (P2) of stochastic dominance (see Section \ref{sec:math}), we have $\mathbb{E}[\max_i Y_i]\le \mathbb{E}[\max_i Y'_i]$. Thus, it suffices to show $\mathbb{E}[\max_i Y'_i]=O(\log N)$ as $N\rightarrow\infty$. Note that $F(y)$ satisfies the conditions in Lemma \ref{lm:asy-order}. Hence, we have $\mathbb{E}[\max_i Y'_i]=m'_N(1+O(1))$ where
\[
\begin{split}
m'_N &= \inf\Big\{y\ge 0:1-F(y)\le \frac{1}{N}\Big\}=\frac{\log N}{\gamma}+\mu.
\end{split}
\]
Since $\mu=O(1)$, we can conclude that $\mathbb{E}[\max_i Y'_i]=O(\frac{\log N}{\gamma})$ as $N\rightarrow\infty$. Now it remains to show that $\gamma=O(1)$ as $N\rightarrow\infty$.
Taking the Taylor expansion of $f_i(s)$ around $s=0$, we have
\begin{equation}\label{eq:fi}
f_i(s)=\sum_{n=0}^{\infty}\frac{f_i^{(n)}(0)}{n!}s^n= (1-\frac{\mu_i}{y})s-\frac{1}{y}\sum_{n=2}^{\infty}\frac{\kappa_n}{n!}s^n,
\end{equation}
where $\kappa_n$ is the $n$-th cumulant\footnote{The $n$-th cumulant of $Y_i$ is the $n$-th order derivative for the logarithm of the MGF of $Y_i$, evaluated at zero, i.e., $\kappa_n=K^{(n)}(0)$, where $K(s)=\log M_{Y_i}(s)$.} of $Y_i$. Note that $\kappa_n$ is a degree-$n$ polynomial in the first $n$ moments of $Y_i$. Since $\mathbb{E}[Y_i^n]=O(1)$ for all $n\in\mathbb{Z}^+$, we have $\kappa_n=O(1)$ as $N\rightarrow\infty$. As a result, the second term in \eqref{eq:fi} (i.e., $\frac{1}{y}\sum_{n=2}^{\infty}\frac{\kappa_n}{n!}s^n$) is also independent of $N$. Hence, there exists some $\epsilon_i$ independent of $N$ such that $\gamma_i=f_i(\epsilon_i)>0$ and $\gamma_i$ is independent of $N$, which implies that $\gamma=\max_i\gamma_i=O(1)$.
\end{proof}
\vspace{-3mm}

\noindent \hrulefill\\
\textbf{Takeaway.} Note that Lemma \ref{lm:asy-order-light} only shows the scaling of $\mathbb{E}[\max_i Y_i]$ in the case where all the moments of $Y_i$ are constants as compared to $N$. Sometimes, we are also interested in the case where $\mathbb{E}[Y_i]=O(f(N))$ and $f(N)\rightarrow\infty$ as $N\rightarrow\infty$. For simplicity, we assume that $Y_i$'s are i.i.d. random variables.



\begin{corollary}\label{coro:taylor-scaling}
Suppose $Y_1,\cdots,Y_N$ are i.i.d. light-tailed random variables with $\mathbb{E}[Y_i]=O(f(N))$ where $f(N)\rightarrow\infty$ as $N\rightarrow\infty$. If $\mathbb{E}[Y_i^n]=O(f^n(N))$ for all $n\in\mathbb{Z}^+$, then $\mathbb{E}[\max_i Y_i]=O(f(N)\log N)$ as $N\rightarrow\infty$.
\end{corollary}
\noindent The proof to this corollary is omitted for brevity since it is similar to the proof of Lemma \ref{lm:asy-order-light} except that we explicitly set $\epsilon_i=\frac{1}{f(N)}$.
\vspace{-1mm}

\noindent \hrulefill

\vspace{1mm}
With Lemma \ref{lm:asy-order-light}, we can easily prove the theorem. It is clear that
\[
\begin{split}
\mathbb{E}[\tau(\mathbf{X})]&\le \mathbb{E}[ \max_i \sum_{j} X_{ij}]+ \mathbb{E}[ \max_j \sum_{i} X_{ij}].
\end{split}
\]
By our assumption, $\sum_{j} X_{ij},~i=1,\cdots,N$ is a sequence of independent light-tailed random variables with $\mathbb{E}[(\sum_j X_{ij})^n]=O(1)$ as $N\rightarrow\infty$ for all $n\in\mathbb{Z}^+$. By Lemma \ref{lm:asy-order-light}, we have
$
\mathbb{E}[\max_i \sum_{j} X_{ij}]=O(\log N).
$
Similarly, we have $\mathbb{E}[\max_j \sum_{i} X_{ij}]=O(\log N)$. As a result, we can conclude that $\mathbb{E}[\tau(\mathbf{X})]=O(\log N)$.

To show the tightness, we consider a scenario where $\mathbf{X}$ is a diagonal matrix: $X_{ij}=0$ with probability 1 for $i\ne j$ and $X_{ii}$ has a geometric distribution with mean $\beta$ for all $i\in[N]$. In this case, we have $\mathbb{E}[\tau(\mathbf{X})]= \mathbb{E}[\max_i X_{ii}]$. It was shown in \cite{geo-max} that the expectation of the maximum of $N$ i.i.d. geometric random variables is $\Theta(\log N)$ as $N\rightarrow\infty$. As a result, $\mathbb{E}[\tau(\mathbf{X})]=\Theta(\log N)$ in this scenario.
\subsection{Proof to Theorem \ref{thm:random-delay}}\label{appendix:random}
We only prove the result for the periodic scheduling policy; the randomized policy can be analyzed in exactly the same way. We assume the system is initially empty. Coflow arrivals are indexed by $k=1,2,\cdots$. Suppose the traffic matrix of coflow $k$ is $\mathbf{X}^{(k)}=(X_{ij}^{(k)})$, and denote by $T^{(k)}$  the inter-arrival time between coflow $k$ and coflow $k+1$. Let $W^{(k)}_{ij}$ be the queuing delay experienced by coflow $k$ in queue $(i,j)$, and denote by $U_{ij}^{(k)}$  the processing time for the batch of coflow $k$ in queue $(i,j)$. Under periodic scheduling, each queue gets served every $N$ time slots. As a result, we have
\[
U_{ij}^{(k)}=\Big(Y_{ij}^{(k)}+N(X^{(k)}_{ij}-1)\Big)^+,
\]
where $(x)^+ = \max(0,x)$ and $Y_{ij}^{(k)}$ is the processing time of the first packet of coflow $k$ in queue $(i,j)$. It is clear that $Y_{ij}^{(k)}\le N$ with probability 1 under periodic scheduling. As a result, we have $U_{ij}^{(k)}\le NX_{ij}^{(k)}\triangleq \tilde{U}_{ij}^{(k)}$ with probability 1. Obviously, the delay performance of the original system (with the batch processing time $U_{ij}^{(k)}$) is upper-bounded by the delay performance under a system where the batch processing time is $\tilde{U}_{ij}^{(k)}$. The latter system is referred to as ``System 2", and denote by $\tilde{W}_{ij}^{(k)}$ the queuing delay experienced by coflow $k$ in queue $(i,j)$ in System 2. Now we show that $\tilde{W}_{ij}^{(k)}$'s are \emph{associated random variables} (see Section \ref{sec:math} for the definition). First, some simple associated random variables are identified in our context.

\begin{lemma}\label{lm:simple-associate}
Define $V^{(k)}_{ij}\triangleq NX^{(k)}_{ij}-T^{(k)}$. For any $k\in\mathbb{Z}^+$, random variables $V^{(k)}_{ij},~i,j\in[N]$ are associated.
\end{lemma}
\begin{proof}
Given $T^{(k)}=t$, it is clear that $V^{(k)}_{ij}$'s are independent and thus associated (by (P1) of associated random variables). As a result, by the definition of associated random variables, for any non-decreasing functions $f$ and $g$
\[
\text{Cov}\Big(f(\mathbf{V}^{(k)}),g(\mathbf{V}^{(k)})\Big|T^{(k)}=t\Big)\ge 0.
\]
As a result, it follows that
\[
\begin{split}
&\text{Cov}\Big(f(\mathbf{V}^{(k)}),g(\mathbf{V}^{(k)})\Big)\\
=&\mathbb{E}\Big[\text{Cov}\Big(f(\mathbf{V}^{(k)}),g(\mathbf{V}^{(k)})\Big|T^{(k)}\Big)\Big]\\
=&\int_{t=0}^{\infty} \text{Cov}\Big(f(\mathbf{V}^{(k)}),g(\mathbf{V}^{(k)})\Big|T^{(k)}=t\Big) f_{T^{(k)}}(t)dt \ge 0,
\end{split}
\]
where $f_{T^{(k)}}(t)$ is the PDF for $T^{(k)}$. Therefore, random variables $V^{(k)}_{ij},~i,j\in[N]$ are associated, and this conclusion holds for all $k\in\mathbb{Z}^+$.
\end{proof}
With Lemma \ref{lm:simple-associate}, we can show that random variables $\tilde{W}_{ij}^{(k)}$'s are associated.

\begin{lemma}\label{lm:waiting-associate}
For all $k\in\mathbb{Z}^+$, random variables $\tilde{W}_{ij}^{(k)},~i,j\in[N]$ are associated.
\end{lemma}
\begin{proof}
We prove by induction on $k$.

\textbf{Basis Step.} When $k=1$, random variables $\tilde{W}_{ij}^{(1)},~i,j\in[N]$ are clearly associated since the system is initially empty and $\tilde{W}_{ij}^{(1)}=0$ with probability 1 for all $i,j\in[N]$.

\textbf{Inductive Step.} For some $k\ge 1$, assume that random variables $\tilde{W}_{ij}^{(k)},~i,j\in[N]$ are associated. Now we prove that random variables $\tilde{W}_{ij}^{(k+1)},~i,j\in[N]$ are also associated. It is clear that
\[
\begin{split}
\tilde{W}_{ij}^{(k+1)} &= (\tilde{W}_{ij}^{(k)} + \tilde{U}_{ij}^{(k)} - T^{(k)})^+\\
& = (\tilde{W}_{ij}^{(k)} + NX^{(k)}_{ij} - T^{(k)})^+.
\end{split}
\]
We identify three sets of associated random variables:
\begin{itemize}
\item[(1)] Random variables $NX^{(k)}_{ij} - T^{(k)},~i,j\in[N]$ are associated due to Lemma \ref{lm:simple-associate};

\item[(2)] The union of $\{\tilde{W}_{ij}^{(k)},~i,j\in[N]\}$ and $\{NX^{(k)}_{ij} - T^{(k)},~i,j\in[N]\}$ is a set of associated random variables since they are two sets of independent associated random variables (by (P3) of associated random variables);

\item[(3)] Random variables $\tilde{W}_{ij}^{(k+1)},~i,j\in[N]$ are associated since they are non-decreasing functions of the set of associated random variables identified in (2) (by (P2) of associated random variables).
\end{itemize}
Now we have completed the induction.
\end{proof}

Since the above lemma holds for all $k\in\mathbb{Z}^+$, we drop the superscript $k$ for convenience. Define $\tilde{Z}_{ij} = \tilde{W}_{ij} 1_{\{X_{ij}\ge 1\}}+ NX_{ij}$, i.e., the total time that coflow $k$ needs to wait for until its batch in queue $(i,j)$ is cleared. Note that the $1_{\{X_{ij}\ge 1\}}$ term is due to the fact that if a coflow contains no packets in queue $(i,j)$, it does not need to experience the queuing delay $\tilde{W}_{ij}$. By Lemma \ref{lm:waiting-associate}, $\tilde{W}_{ij}$'s are associated;  $X_{ij}$'s are also associated due to the independence (P1); the union of $\tilde{W}_{ij}$'s and $X_{ij}$'s is also a set of associated random variables due to the independence of $X_{ij}$'s and $\tilde{W}_{ij}$'s (P3). Therefore, we can conclude that $\tilde{Z}_{ij}$'s are associated since they are non-decreasing functions of $\tilde{W}_{ij}$'s and $X_{ij}$'s (P2).
Note that $\max_{ij} \tilde{Z}_{ij}$ is the coflow-level delay experienced by a coflow in System 2. By (P4) of associated random variables, we have $\mathbb{E}[\max_{ij}\tilde{Z}_{ij}]\le \mathbb{E}[\max_{ij}\tilde{Z}'_{ij}]$ where $\tilde{Z}'_{ij}$'s are \emph{independent} random variables identically distributed as $\tilde{Z}_{ij}$'s. In order to evaluate the value of $\mathbb{E}[\max_{ij}\tilde{Z}'_{ij}]$, we identify some important properties of the distribution of $\tilde{Z}_{ij}$.

First, we claim that $\tilde{Z}_{ij}$'s are light-tailed. Indeed, since $X_{ij}$ is light-tailed, the service time $\tilde{U}_{ij}=NX_{ij}$ is also light-tailed. According to \cite{gg1-tail} (Theorem 11, p. 129), if the service time is light-tailed, the tail of the queuing time in an M/G/1 queue decreases exponentially. Hence, $\tilde{W}_{ij}$'s are light-tailed, which implies that $\tilde{Z}_{ij}$'s also have light-tailed distributions.

Next, we evaluate the value of $\mathbb{E}[\tilde{Z}_{ij}]$ and higher moments of $\tilde{Z}_{ij}$. It is clear that queue $(i,j)$ is a slotted M/G/1 queue with arrival rate $\lambda$ and service time $\tilde{U}_{ij} =NX_{ij}$. Note that $\mathbb{E}[\tilde{U}_{ij}] = \beta$ and $\mathbb{E}[\tilde{U}^2_{ij}]=(N\sigma^2+\beta^2)$. According to the Pollaczek-Khinchin formula for slotted M/G/1 queues, we have
\[
\begin{split}
\mathbb{E}[\tilde{W}_{ij}] &= \frac{\lambda \mathbb{E}[\tilde{U}_{ij}^2]}{2(1-\lambda\mathbb{E}[\tilde{U}_{ij}])}+\frac{1}{2}\\
& =\frac{\lambda (N\sigma^2+\beta^2)}{2(1-\rho)}+\frac{1}{2}.
\end{split}
\]
As a result, we have
\[
\begin{split}
\mathbb{E}[\tilde{Z}_{ij}] &= \mathbb{P}(X_{ij}\ge 1)\mathbb{E}[\tilde{W}_{ij}]+\mathbb{E}[\tilde{U}_{ij}]\\
& \le \frac{\beta}{N}\Big[\frac{\lambda (N\sigma^2+\beta^2)}{2(1-\rho)}+\frac{1}{2}\Big]+\beta\\
&= O(1),
\end{split}
\]
where the inequality utilizes the Markov inequality, i.e., $\mathbb{P}(X_{ij}\ge 1)\le \mathbb{E}[X_{ij}]=\frac{\beta}{N}$.  Furthermore, noting that $\mathbb{E}[\tilde{U}_{ij}^n]=N^n\mathbb{E}[X^{n}_{ij}]=O(N^{n-1})$ for all $n\in\mathbb{Z}^+$ (by our assumptions on $X_{ij}$'s), we can similarly obtain $\mathbb{E}[\tilde{Z}_{ij}^n]=O(N^{n-1})$ according to the moment bounds in \cite{gg1-tail}.

Recall that $\tilde{Z}'_{ij}$'s are independent random variables identically distributed as $\tilde{Z}_{ij}$'s.  Define $\tilde{Z}'_i\triangleq \max_j \tilde{Z}'_{ij}$ for all $i\in[N]$. It follows from the above analysis that $\tilde{Z}'_i,~i=1,\cdots,N$ are a sequence of i.i.d. light-tailed random variables with $\mathbb{E}[\tilde{Z}'_i]\le \sum_j\mathbb{E}[\tilde{Z}'_{ij}]=O(N)$ and $\mathbb{E}[(\tilde{Z}'_i)^n]\le O\big(\sum_j \mathbb{E}[(\tilde{Z}'_{ij})^n]\big)=O(N^n)$ for all $n\in\mathbb{Z}^+$. By Corollary \ref{coro:taylor-scaling} (see Appendix \ref{appendix:ct-light}), we have $\mathbb{E}[\max_i \tilde{Z}'_i]=O(N\log N)$ as $N\rightarrow\infty$. Then it follows that
\[
\mathbb{E}[\max_{i,j}\tilde{Z}_{ij}]\le \mathbb{E}[\max_{i,j}\tilde{Z}'_{ij}]=\mathbb{E}[\max_i \tilde{Z}'_{i}] =O(N\log N).
\]
Therefore, the average coflow-level delay achieved by the periodic scheduling policy is $O(N\log N)$.

\subsection{Flow size distribution in an overflow frame}\label{app:overflow-size}
In this appendix, we introduce a lemma about the flow size distribution of a coflow that arrives in an \textbf{overflow} frame.
\begin{lemma}\label{lm:non-conforming-batch}
Let $(\tilde{X}_{ij})$ be the traffic of an arbitrary coflow in an \textbf{overflow} frame. Then $\sum_j \tilde{X}_{ij}$ is stochastically dominated by $\sum_j X_{ij}+T$ for all $i\in[N]$, where $(X_{ij})$ is the traffic MATRIX of an arbitrary coflow in an \textbf{arbitrary} frame.
\end{lemma}
\begin{proof}
Let $\mathbf{L}$ be the total traffic that arrives in an arbitrary frame, and denote by $\mathbf{X}=(X_{ij})$ the traffic associated with an arbitrary coflow in this frame. Clearly, an overflow occurs in this frame if and only if $\tau(\mathbf{L})\ge T$. As a result, the flow size distribution for an arbitrary coflow in this overflow frame is
\[
\mathbb{P}[X_{ij}\ge m|\tau(\mathbf{L})\ge T]\triangleq\mathbb{P}[\tilde{X}_{ij}\ge m],
\]
where $(X_{ij})$ is the traffic of an arbitrary coflow in an arbitrary frame.

In order to show that $\sum_{j} \tilde{X}_{ij}$ is stochastically dominated by $\sum_j X_{ij}+T$ for all $i\in[N]$, it suffices to show that for all $i\in[N]$
\[
\small
\mathbb{P}[\tilde{X}_{ij}\ge m]\le \mathbb{P}[\sum_j X_{ij}+T\ge m].
\]
To show this inequality, we first prove that for all $i\in[N]$
\begin{equation}\label{eq:batch-no}
\small
\mathbb{P}[\sum_j X_{ij}\ge m|\tau(\mathbf{L})\ge T] \le \mathbb{P}[\sum_j X_{ij}\ge m|\sum_j X_{ij}\ge T].
\end{equation}
If $m< T$, we note that the right-hand side of \eqref{eq:batch-no} equals to 1, so inequality \eqref{eq:batch-no} naturally holds. If $m\ge T$, we notice that
\[
\small
\sum_j X_{ij}\ge m\Rightarrow \tau(\mathbf{L})\ge \sum_j X_{ij}\ge m\ge T,
\]
which implies that
\begin{equation}\label{eq:11}
\small
\mathbb{P}[\sum_j X_{ij}\ge m,\tau(\mathbf{L})\ge T] = \mathbb{P}[\sum_j X_{ij}\ge m].
\end{equation}
Similarly, we can show that when $m\ge T$
\begin{equation}\label{eq:12}
\small
\mathbb{P}[\sum_j X_{ij}\ge m, \sum_j X_{ij}\ge T]=\mathbb{P}[\sum_j X_{ij}\ge m].
\end{equation}
Comparing \eqref{eq:11} with \eqref{eq:12}, we have when $m\ge T$
\begin{equation}\label{eq:13}
\small
\mathbb{P}[\sum_j X_{ij}\ge m, \tau(\mathbf{L})\ge T]=\mathbb{P}[\sum_j X_{ij}\ge m, \sum_j X_{ij}\ge T].
\end{equation}
Now we rewrite $\small \mathbb{P}[\sum_j X_{ij}\ge m|\tau(\mathbf{L})\ge T]$ according to the rule of conditional probability:
\[
\small
\begin{split}
\mathbb{P}[\sum_j X_{ij}\ge m|\tau(\mathbf{L})\ge T]&=\frac{\mathbb{P}[\sum_j X_{ij}\ge m,\tau(\mathbf{L})\ge T]}{\mathbb{P}[\tau(\mathbf{L})\ge T]}\\
&\le\frac{\mathbb{P}[\sum_j X_{ij}\ge m, \sum_j X_{ij}\ge T]}{\mathbb{P}[\sum_j X_{ij}\ge T]}\\
&=\mathbb{P}[\sum_j X_{ij}\ge m|\sum_j X_{ij}\ge T],
\end{split}
\]
where the inequality is due to \eqref{eq:13} and the fact that $\mathbb{P}[\tau(\mathbf{L})\ge T]\ge \mathbb{P}[\sum_j X_{ij}\ge T]$. This completes the proof to \eqref{eq:batch-no}.

Since $\sum_j X_{ij}$ is light-tailed, we know from \cite{gallager-stochastic} that as $T\rightarrow\infty$,
\begin{equation}\label{eq:gallager}
\small
\mathbb{P}[\sum_j X_{ij}\ge m|\sum_j X_{ij}\ge T]\le \mathbb{P}[\sum_j X_{ij} \ge m-T].
\end{equation}
Taking \eqref{eq:gallager} into \eqref{eq:batch-no}, we obtain
\begin{equation}\label{eq:eq-haha}
\small
\mathbb{P}[\sum_j X_{ij}\ge m|\tau(\mathbf{L})\ge T]\le \mathbb{P}[\sum_j X_{ij} \ge m-T].
\end{equation}
Note that the left-hand side of \eqref{eq:eq-haha} equals to $\mathbb{P}[\sum_j \tilde{X}_{ij}\ge m]$ by our definition of $\tilde{X}_{ij}$, and the right-hand side of \eqref{eq:eq-haha} can be rewritten as $\mathbb{P}[\sum_j X_{ij}+T\ge m]$. As a result, we have for all $i\in[N]$
\[
\mathbb{P}[\sum_j \tilde{X}_{ij}\ge m]\le \mathbb{P}[\sum_j X_{ij}+T\ge m],
\]
i.e.,  $\sum_j \tilde{X}_{ij}$ is stochastically dominated by $\sum_j X_{ij}+T$.
\end{proof}
\end{document}